\documentclass[runningheads]{llncs}
\usepackage{cmap}
\usepackage[T1]{fontenc}
\usepackage[utf8]{inputenc}
\usepackage{lmodern}
\usepackage{microtype}
\usepackage{graphicx}
\usepackage[cmex10]{amsmath}
\usepackage{amssymb}
\usepackage{booktabs}
\usepackage{array}
\usepackage{multirow}
\usepackage{threeparttable}
\usepackage{subcaption}
\usepackage{float}
\usepackage{stfloats}
\usepackage{placeins}
\usepackage{indentfirst}
\usepackage[linesnumbered,ruled,vlined]{algorithm2e}
\usepackage{hyperref}
\hypersetup{
  hidelinks,
  pdftitle={BRASP: Hiding Where You Look in Encrypted Spatial Keyword Search},
  pdfauthor={Jing Zhang; Ganxuan Yang; Yifei Yang; Siqi Wen; Baiyan Li; Zhengyang Qiu; Hasan Ahmed; Muhammad Bilal},
  pdfkeywords={Searchable encryption, Spatial queries, Private information retrieval, Access-pattern privacy, Forward privacy}
}

\begin{document}
\raggedbottom
\title{BRASP: Hiding Where You Look in Encrypted Spatial Keyword Search}
\titlerunning{BRASP}
\author{Jing Zhang \and Ganxuan Yang \and Yifei Yang \and Siqi Wen \and
Baiyan Li \and Zhengyang Qiu \and Hasan Ahmed \and Muhammad Bilal}
\authorrunning{Jing Zhang et al.}
\institute{Lancaster University\\[0.5ex]
Code repository: \url{https://github.com/Egbert-Lannister/BRASP}}
\maketitle

\begin{abstract}
Spatial keyword queries over outsourced encrypted data can reveal correlations between queries through repeated accesses, even when data and query contents remain confidential. Existing dynamic searchable encryption schemes may provide forward and backward privacy while still leaking search repetitions and update-related information. We present BRASP, an encrypted query protocol supporting spatial ranges and conjunctive keyword conditions in a setting with two non-colluding semi-honest servers. BRASP uses two-stage private information retrieval to protect both index accesses and candidate-object reads, followed by local filtering on original coordinates to obtain exact results. Queries leave the cloud-side snapshot unchanged. New objects are buffered at a trusted client and published through batch commits that rebuild the complete index and object tables with fresh encryption and permutations, hiding the locations affected by insertions. Each query uses a single committed database version, and inserted objects become queryable after commit completion. We prove query correctness and privacy against either server individually over histories of queries and insertion commits, subject to explicit leakage of read counts and public scheduling. Experiments on the Yelp dataset show lower construction and search times than the evaluated $\mathrm{VPBRQ}_{SupL}$ and PPSKS implementations under the tested configurations.

\end{abstract}

\section{Introduction}
Spatial keyword queries retrieve objects satisfying both spatial range and keyword conditions, such as businesses within a specified region whose records contain both ``coffee'' and ``parking.'' Cloud outsourcing offers scalable storage and computation while reducing maintenance burdens. However, content encryption alone does not hide access behavior: visible read locations can reveal whether the same encrypted records are read repeatedly or accessed by different queries. Protecting such services therefore requires controlling access-pattern leakage during query execution as well as keeping data and query contents confidential.

Prior work on encrypted spatial data provides several forms of privacy protection. Wang et al.'s $\mathrm{DSSE}_{\mathsf{SKQ}}$~\cite{wang2022forward} provides forward, backward, and content privacy, but its leakage definition exposes repetitions of keyword or spatial-prefix searches and associated update timestamps. Liu et al.'s REACTS~\cite{liu2026reacts} further addresses update security and robustness, but its stated search leakage includes search repetitions, permuted identifiers of matching objects, and insertion timestamps. Figure~\ref{fig:motivation}(a) illustrates how linkable search information reveals repetitions without disclosing the underlying predicate. Forward and backward privacy alone therefore do not hide such correlations. Tong et al.~\cite{tong2024beyond} protect search and access patterns in spatial Boolean range queries and provide result verification. For a service that continually accepts queries and insertions, we ask how to protect index retrieval, candidate-record delivery, and insertion maintenance jointly across the operation history, with an explicit account of what remains observable.

Candidate retrieval and updates make this question concrete. A cell-based index represents spatial membership at cell granularity. When a query rectangle cuts through a boundary cell, the index may also select objects whose original coordinates lie outside the rectangle. The client must retrieve candidate records to determine exact spatial membership. In Figure~\ref{fig:motivation}(b), $Q_1$ reads the records in Slots 2 and 7, then excludes the record in Slot 7 using its original coordinates. If these read locations are visible, the server can recognize that $Q_2$ also reads Slot 7, even though $Q_1$ does not return that record. This overlap links queries through data absent from one query's final answer: exact filtering does not erase the relationship exposed by candidate reads. Insertions create a corresponding write-side issue, because storing a new record also changes its spatial and keyword index entries. Figure~\ref{fig:motivation}(c) illustrates sparse writes to $I[2]$, $I[5]$, and $O[3]$; encrypting their contents still leaves the modified locations visible. These examples motivate protection across index reads, candidate reads, and insertion-induced changes.

\begin{figure}[tbp]
\centering
\includegraphics[width=0.90\linewidth]{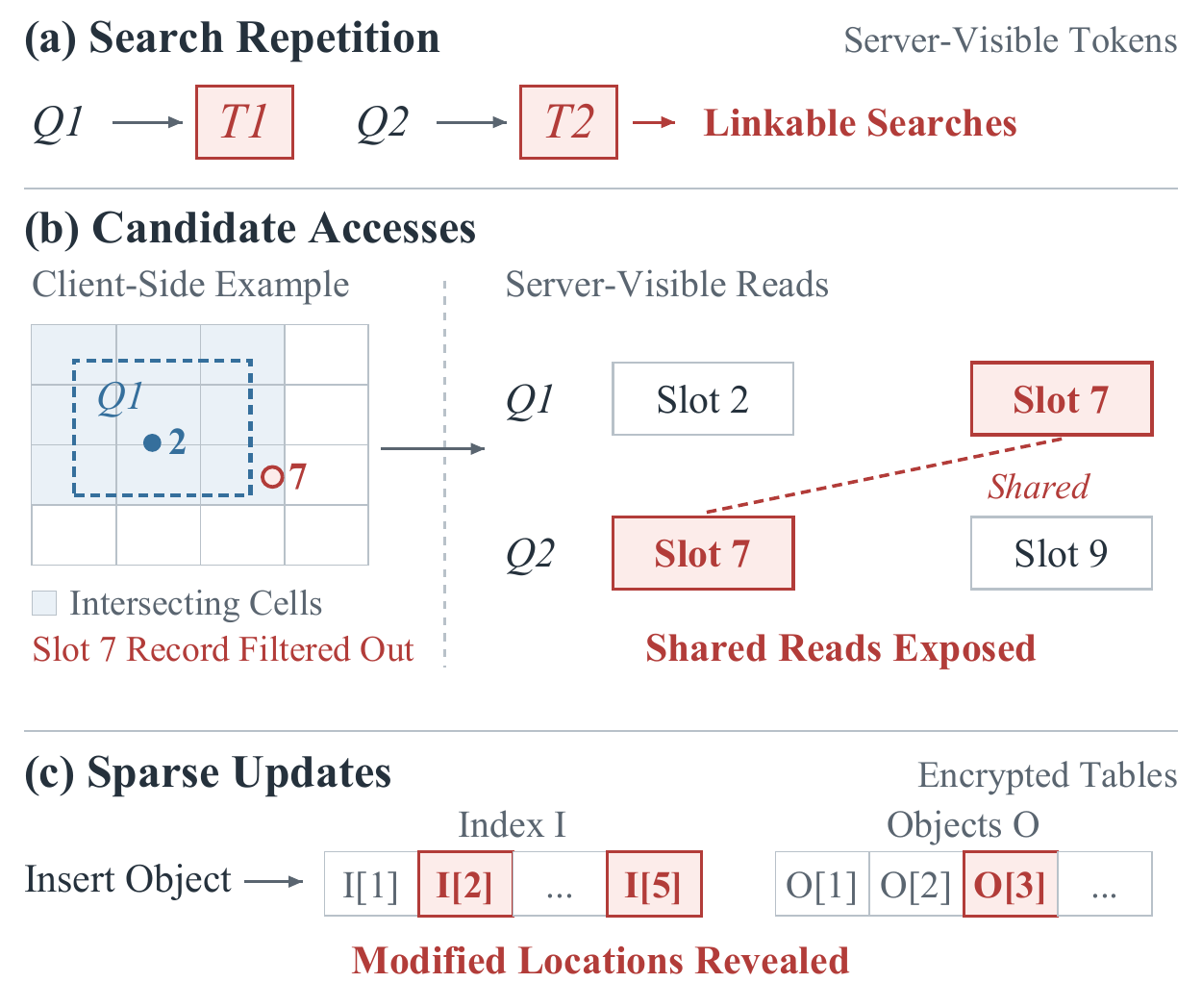}
\caption{Potential privacy leakage from linkable searches, direct candidate reads, and sparse updates.}
\label{fig:motivation}
\ifdefined\Description
\Description{Three illustrative leakage cases. Distinct but linkable search tokens T1 and T2 expose repeated predicates. A client-side grid shows a record in Slot 7 that is read but later excluded from Q1, while visible read traces show Q1 and Q2 both reading Slot 7. Sparse writes expose modified locations I[2], I[5], and O[3]. Red marks potential leakage; the blue region shows the client-side query and intersecting cells. These examples are not BRASP transcripts or defects attributed to all cited schemes.}
\fi
\end{figure}

We propose BRASP, a protocol for spatial range queries with conjunctive keyword conditions, addressing these three observation points together. Two non-colluding semi-honest servers store identical encrypted index and object tables and communicate with the trusted client over separate private authenticated channels. Two-stage private information retrieval (PIR)~\cite{chor1998private} first retrieves spatial and keyword bitmaps to identify candidates, then retrieves their encrypted records for exact coordinate filtering at the client. Privacy at the index layer does not automatically extend to object retrieval: requesting a candidate by its ordinary storage address would reveal a target that the private index lookup had concealed. Fresh randomized requests for every read hide the targets from either server individually, addressing both the predicate linkage in Figure~\ref{fig:motivation}(a) and the shared Slot 7 access in Figure~\ref{fig:motivation}(b). For insertions, the client buffers new objects and commits them in batches. Each commit rebuilds and freshly encrypts the complete index and object tables under fresh secret permutations. Complete replacement hides which entries the insertions affect, instead of exposing the local modifications illustrated in Figure~\ref{fig:motivation}(c). Each read-only query uses one committed version throughout; buffered objects become queryable only after commit completion. These protections are relative to explicit leakage: index- and object-read counts, public scheduling, and commit batch sizes remain observable.

Our main contributions are as follows: (a) We propose BRASP, combining two-stage private retrieval and batch insertion commits to protect index and candidate-object access patterns and hide the storage locations affected by insertions. (b) We prove query correctness and security across query and insertion-commit histories under an explicit leakage profile. (c) Experiments on Yelp show lower construction and search times than the evaluated $\mathrm{VPBRQ}_{SupL}$ and PPSKS implementations under the tested configurations.

\section{Related Work}\label{sec:related}

Privacy-preserving spatial queries encompass several combinations of location predicates, descriptive keywords, and security requirements. Miao et al.'s survey~\cite{miao2023comprehensive} organizes this literature by query functionality, index structure, complexity, and security. These dimensions are useful for positioning BRASP: support for a spatial predicate, protection of query contents, and concealment of the resulting storage accesses are separate properties.

Encrypted geometric search provides foundations for spatial queries. FastGeo~\cite{wang2017fastgeo} supports arbitrary geometric regions and dynamic updates, while Xu et al.~\cite{xu2018enabling} integrate spatial-data access control with geometric range search. Kermanshahi et al.~\cite{kermanshahi2020geometric} develop dynamic geometric search with forward/backward security and content privacy for updated points. Sun et al.~\cite{sun2023towards} study high-dimensional range queries with dimension selection, and Liang et al.~\cite{liang2024efficient} combine range encoding with hidden-vector encryption and a tree-indexed extension.

Index organization and update privacy remain active design dimensions. Miao et al.~\cite{miao2023efficient} develop PSDQ and an enhanced variant using a unified index and Geohash-based R-tree pruning. Shang et al.~\cite{shang2022privacy} examine private multidimensional range queries for edge-supported industrial IoT. Li et al.~\cite{li2025multilevel} study multilevel matching for encrypted geometric range queries with forward and backward privacy. BRASP uses spatial encoding to obtain candidates, then treats their encrypted-record retrieval as a separate private operation before exact coordinate filtering.

Spatial-keyword encryption already supports the predicate considered here. Cui et al.~\cite{cui2019geo} combine a spatial-textual Bloom filter with an encrypted R-tree; Wang et al.'s PBRQ-L/Q~\cite{wang2020search} support spatial ranges and conjunctive keywords under an explicit leakage model. Wang et al.~\cite{wang2021enabling} also formalize spatial keyword structured encryption with linear and tree-based constructions and post-validation of candidates. Zhao et al.~\cite{zhao2024towards} combine Hilbert encoding with a quadtree-based secure index, while RASK/RASK+~\cite{lv2023rask} use symmetric encryption and keyed hashing for range spatial-keyword queries. Ge et al.~\cite{ge2024privacy} combine forward-secure updates with result verification for Boolean range queries.

Further constructions study geographic keyword Boolean range queries~\cite{gong2022efficient}, spatial-keyword range processing in the cloud through EPSRQ~\cite{jiang2025epsrq}, and multidimensional ranges combined with keyword search through PMRK~\cite{tu2024pmrk}. This body of work establishes that combining spatial and textual predicates is already a well-developed query capability. BRASP's contribution concerns the joint treatment of index retrieval, candidate-object delivery, and insertion commits under one explicit history-level leakage model.

Tong et al.~\cite{tong2024beyond} are particularly close: their spatial-keyword construction uses distributed point functions and pseudorandom-permutation-based Cuckoo hashing to protect search and access patterns while verifying results. BRASP focuses on integrating private index and candidate-object reads, exact spatial filtering, and insertion commits within one history-based analysis; result verification is separate from its semi-honest-server guarantee. Zhang et al.'s PPSKS~\cite{zhang2022efficient} instead addresses spatial keyword similarity, a different predicate from conjunction. Our evaluated PPSKS and $\mathrm{VPBRQ}_{SupL}$ implementations provide performance context; Table~\ref{tab:related-comparison} separately compares selected functional and security properties.

The similarity-search literature further includes Song et al.'s combination of geometric ranges, keyword similarity, and role-based access control~\cite{song2021privacy}; Zhang et al.'s secure similarity queries with arbitrary spatial ranges~\cite{zhang2024performance}; and PPSKSQ~\cite{wang2025ppsksq}, which retrieves points with high keyword-set similarity within a query range. These predicates require specifying a similarity measure in addition to a spatial condition. Authorization introduces another requirement: Wen et al.'s FPAQ~\cite{wen2025fast} combines Geohash/Quadtree indexing with query authorization and attribute-based access control for multiple users. BRASP instead assumes a trusted client authorized to retrieve the database and studies what the storage servers observe. Extending it to distinct user privileges would require an explicit authorization model.

Dynamic searchable encryption distinguishes update privacy from search leakage. Sophos~\cite{bost2016ovarphiovarsigma} provides forward privacy while disclosing search repetitions and relevant update histories. MoRSE/$(\mathrm{RS})^2$~\cite{wangchow2022range} support one-dimensional ranges and insertion/deletion with leakage concerning range-cover nodes. ODXT~\cite{patranabis2020forward} extends forward/backward-private search to conjunctions, but its stated leakage includes matching identifiers, insertion timestamps, and related update histories. HDXT~\cite{yuan2023result} hides additional keyword-pair intersection identifiers while permitting disclosure of the final conjunction result. It also caches updates and rerandomizes all attribute-map entries during eviction. BRASP commits both spatial index and object tables as a snapshot, hiding affected destinations at the cost of delayed insertion visibility and full-table work; it does not support deletions.

Related dynamic-search constructions improve different aspects of the query service. FAST and FASTIO~\cite{song2018forward} develop forward-private search using symmetric primitives, with FASTIO emphasizing I/O efficiency. VBTree~\cite{wu2019vbtree} supports forward-secure conjunctive queries using a virtual tree and version control, while Zuo et al.~\cite{zuo2020forward} use an extended bitmap index for forward- and backward-private conjunctive search. Verification is another dimension: Guo et al.~\cite{guo2023forward} combine forward privacy with verifiable dynamic conjunctions and describe a backward-private extension; Li et al.~\cite{li2021verifiable} study verifiable dynamic multi-keyword search using bitmaps. These results provide relevant index and update techniques. For BRASP, the separate question is how the transcript covers candidate-record delivery and insertion maintenance in addition to the index operation.

Search relationships also matter beyond update privacy. Oya and Kerschbaum~\cite{oya2021hiding} demonstrate query-recovery attacks combining search/access leakage with auxiliary information. OSSE~\cite{shang2021obfuscated} independently obfuscates each query's access pattern using randomized false positives and false negatives. SAP-SSE~\cite{song2020sap} protects search and access patterns through index shuffling and redistribution across non-colluding clouds. Xie et al.~\cite{xie2024access} use distributed point functions for access-pattern-hiding keyword search and updates. BRASP targets exact spatial-keyword queries with private index and object reads and explicit insertion-commit semantics under its stated leakage and trust assumptions.

Search-pattern privacy is also an explicit objective in the schemes of Zheng et al.~\cite{zheng2020achieving} and Wang et al.~\cite{wang2020achieving}. OpenSE~\cite{yang2024opense} combines search- and access-pattern hiding with verifiable keyword search against malicious server behavior, using oblivious polynomial evaluation. MFSSE~\cite{chen2024mfsse} addresses fuzzy multi-keyword ranked search with randomized trapdoors and response perturbation. Its approximate matching and perturbed responses differ from BRASP's exact conjunction and local spatial filtering. These works show why comparisons must specify the query predicate, observable information, and adversarial model separately. Gui et al.~\cite{gui2023rethinking} further demonstrate that leakage during document retrieval can undermine protections at the searchable-index layer. This system-level perspective directly motivates protecting both index and object reads in BRASP, while retaining an explicit account of the visible request counts.

At the storage layer, Chor et al.'s two-server PIR~\cite{chor1998private} hides read addresses from either replica. BRASP uses the elementary XOR-selector construction and accounts for its full scans and linear-size requests. Distributed point functions~\cite{gilboa2014distributed} offer compact query representations. Path ORAM~\cite{stefanov2013path} instead hides block reads and writes for equal-length access sequences using a private position map and stash; applying it to spatial queries requires specifying the index/object layout and access counts. BRASP combines private reads with complete insertion commits under explicit volume leakage. Section~\ref{sec:phase-costs} details its storage and update tradeoffs.

Related work on spatial and embodied data spans representation, generation, and behavioral safety. Depth-NeuS~\cite{wen2024depthneus} uses depth supervision and geometric consistency constraints for multi-view neural surface reconstruction. RoboImagine~\cite{wang2026roboimagine} generates robotic task videos with geometric and dynamic consistency augmentation, while SAFE-Dict~\cite{wen2026safedict} attenuates harmful concept components in latent representations to improve inference-time safety in vision-language-action models. BRASP addresses a complementary storage-layer problem: protecting index and object accesses for encrypted records containing coordinates and keywords, under an explicit leakage model. Appendix~\ref{app:extended-related} discusses other query operators, storage guarantees, and ciphertext-transformation primitives that further clarify the scope of the protocol.

\begin{table}[tp]
\caption{Comparison of the named constructions across query, leakage, update, and trust-model dimensions.}
\label{tab:related-comparison}
\centering
\begingroup
\footnotesize
\setlength{\tabcolsep}{3pt}
\renewcommand{\arraystretch}{1.18}
\begin{tabular}{@{}>{\raggedright\arraybackslash}p{.165\linewidth}>{\raggedright\arraybackslash}p{.18\linewidth}>{\raggedright\arraybackslash}p{.16\linewidth}>{\raggedright\arraybackslash}p{.195\linewidth}>{\raggedright\arraybackslash}p{.22\linewidth}@{}}
\toprule
Dimension & PBRQ-L/Q~\cite{wang2020search} & Sophos (base)~\cite{bost2016ovarphiovarsigma} & MoRSE/\allowbreak$(\mathrm{RS})^2$~\cite{wangchow2022range} & BRASP \\
\midrule
Query & Spatial range and keyword conjunction & Single keyword & One-dimensional range & Spatial range and keyword conjunction \\
Search pattern & Correlation leakage$^{a}$ & Query repetition revealed & Cover-node searches revealed & Hidden beyond declared leakage$^{b}$ \\
Access pattern & Returned identifiers revealed & Matching update history revealed & Cover-node histories/results revealed & Hidden beyond declared leakage$^{b}$ \\
Updates & Static interface & Insertions & Insertions and deletions & Insertion batches \\
Forward privacy & Not provided & Yes & Yes & At completed commits$^{b}$ \\
Server model & One semi-honest server & One untrusted server & One untrusted server & Two non-colluding semi-honest servers \\
\bottomrule
\end{tabular}

\smallskip
\begin{minipage}{\linewidth}
\footnotesize
$^{a}$ PBRQ's search-pattern definition concerns common returned objects across trapdoors; it is not identical to query equality. Its stated leakage permits that correlation.\\
$^{b}$ BRASP's guarantees are relative to Section~\ref{sec:security}: read counts and public scheduling remain visible. Forward privacy concerns the commit transcript; later query volumes may reveal additional information. The table compares the named constructions and does not imply identical functionality or leakage definitions.
\end{minipage}
\endgroup
\end{table}

\section{Problem Formulation}
\label{sec:problem}

\subsection{System Model}
BRASP has a trusted data owner, a trusted authorized client, and two cloud servers $CS_1,CS_2$ (Figure~\ref{fig:system-model}). The owner initializes the database and grants the client keys and private directories. The client acts as a single trusted coordinator that serializes queries and commits; the owner may perform maintenance if it holds the same state and transfers refreshed directories securely. The baseline coordinator recovers the stored snapshot during maintenance rather than keeping a separate authoritative database.

\begin{figure}[htbp]
\centering
\includegraphics[width=0.90\linewidth]{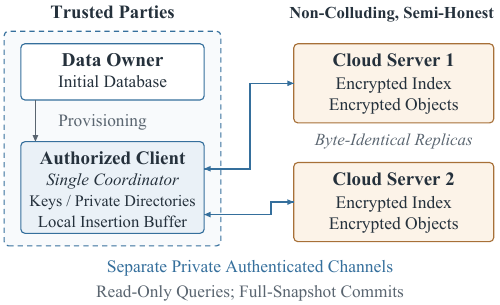}
\caption{System model. Non-colluding servers hold identical encrypted snapshots and communicate privately with the trusted client.}
\label{fig:system-model}
\end{figure}

The public configuration $\mathsf{pp}$ fixes security parameter $\lambda$, spatial domain and cell-boundary convention, Hilbert order $h$, object capacity $N$, index capacity $M$, padded object-plaintext length $\ell_O$, and query/commit policies $\mathsf{pol}_Q,\mathsf{pol}_C$. Index plaintexts have $N$ bits; ciphertext-row lengths are $L_I,L_O$. The object table has $R=N+1$ rows, including logical dummy ID $0$; $M$ includes a designated zero-bitmap row. Padding hides individual-row occupancy.

At epoch $e$, $DB_e$ contains objects $O_j=(id_j,loc_j,W_j,\mathit{payload}_j)$ with immutable internal IDs $id_j=j\in\{1,\ldots,N\}$. IDs remain bitmap coordinates across physical permutations, and records fit $\ell_O$. The client holds independent keys $K_I,K_O$, a predicate-to-index directory $D_e$, zero-row address $d_e$, object permutation $\pi_e:\{0,\ldots,N\}\rightarrow\{1,\ldots,R\}$, live mask $V_e$, persistent nonce/ID allocation state, epoch $e$, and pending buffer $U$. Keywords and prefixes have disjoint namespaces. Both servers store byte-identical snapshots
\[
 S_{i,e}=(e,C_e^I[1..M],C_e^O[1..R]),\qquad i\in\{1,2\},
\]
without keys or directories. Buffered insertions become queryable only after commit completion. Deletions and independent concurrent writers are outside the construction.

Table~\ref{tab:notation} separates logical objects, physical storage, and scheduled reads. An object's immutable identifier names its bitmap coordinate. Its physical object-table slot is instead determined by the private mapping $\pi_e$ and can change at every commit. Index rows have a separate private mapping $D_e$: a keyword and a spatial prefix with the same textual representation still name different predicates. These distinctions let the client combine bitmaps across predicates while both servers see only encrypted arrays and private-read messages.

\begin{table}[htbp]
\caption{Principal notation. Capacities and scheduled reads include padding.}
\label{tab:notation}
\centering
\small
\begin{tabular}{@{}p{.22\linewidth}p{.73\linewidth}@{}}
\toprule
Symbol & Meaning \\
\midrule
$DB_e,U$ & Committed database at epoch $e$ and the trusted insertion buffer \\
$N,R=N+1$ & Capacity for real objects and number of physical object rows, including the dummy \\
$M$ & Number of physical index rows, including zero and padding rows \\
$h,\gamma=2h$ & Hilbert order and encoded cell-value length \\
$D_e,d_e$ & Private predicate directory and designated zero-bitmap slot \\
$\pi_e,V_e$ & Private object-slot mapping and live-ID bitmap \\
$L_I,L_O$ & Fixed encrypted index-row and object-row lengths in bits \\
$a_Q,c_Q$ & Actual predicate count and candidate count before exact filtering \\
$s_Q,v_Q$ & Scheduled index and object reads, including padding \\
$b$ & Number of real insertions published by one commit \\
$D=ML_I+RL_O$ & Payload size of one complete encrypted snapshot in bits \\
\bottomrule
\end{tabular}
\end{table}

\subsection{Threat Model}
Either server may be statically corrupted and retain its auxiliary information, randomness, all past snapshots, messages, and read/write history. Servers follow the protocol and do not collude. Each client--server link is a separate private authenticated channel: one server cannot observe its peer's PIR selections or responses. This requires independently administered operators without shared transcript collection, rather than two virtual machines controlled by one administrator. Combining both selections reveals the physical target, so the guarantee does not cover collusion; address recovery alone does not reveal client-held decryption keys.

Clients and the coordinator are trusted and may recover complete predicate bitmaps. Client compromise, restricted-query authorization against a malicious client, malicious-server responses, rollback, denial of service, and local timing side channels are outside the model. The analysis covers discrete messages under public operation and transmission schedules; query frequency remains visible. Secret contents must not select additional observable scheduling decisions. Workloads fit the fixed capacities, formats, nonce space, and public padding budgets; capacity expansion requires a separately declared public policy.

\subsection{Query Semantics and Interfaces}
A query $Q=(R_q,W_q)$ requests the original records and identifiers
\[
 \Delta_e(Q)=\{id_j:O_j\in DB_e,\ loc_j\in R_q,\ W_q\subseteq W_j\}.
\]
Only the client receives the answer. Index retrieval, candidate-object retrieval, and exact spatial filtering use one pinned epoch. Arbitrary rectangles may select boundary-cell candidates, so the object-read count measures candidates before filtering. The search pattern records equality of predicates; the logical access pattern records result-ID sets over the query history. BRASP hides addresses, repeated identities, intersections, and insertion-to-past-query associations relative to the volume and scheduling leakage in Section~\ref{sec:security}.

The client remains online for queries, which never write back to the servers. The following interfaces share the committed state.

\paragraph{Interfaces.}
\label{app:interfaces}
Setup$(1^\lambda,\mathsf{pp})$ initializes state; BuildSnapshot$(DB,e)$ stages candidates. Reads use PrivateRead$(T,e,j)$ and Search$(Q,e)$. Updates use BufferInsert$(O)$ and CommitEpoch$(e,U)$. Except Setup, interfaces access keys, $\mathsf{pp}$, and counters implicitly; reads use active metadata. The committed state is
\[
 \mathsf{st}_e=(K_I,K_O,e,D_e,d_e,\pi_e,V_e,
                \mathsf{ctr}_I,\mathsf{ctr}_O),
\]
together with $U$ and an ID-allocation counter. Nonce counters are monotone, persistent, and never rolled back after abandoned preparations; the privacy theorem covers only uninterrupted histories. Occupancy is derived from $V_e$.

\section{Preliminaries}
\label{preliminaries}

\subsection{Fixed-Length Authenticated Encryption}
Let $\mathsf{AEAD}=(\mathsf{Gen},\mathsf{Enc},\mathsf{Dec})$ provide correct decryption and multi-message, nonce-respecting IND-CPA confidentiality. Independent keys $K_I,K_O$ encrypt fixed-length index and object records without content-dependent compression. For table $T\in\{\mathsf I,\mathsf O\}$, epoch $e$, physical slot $j$, and padded plaintext $\mu$,
\begin{equation}
 C_e^T[j]=\nu\Vert\mathsf{AEAD.Enc}(K_T,\nu,\mu,(T,e,j)).
 \label{eq:snapshot-row}
\end{equation}
The nonce and tag are included in $L_I,L_O$. Persistent per-key counters allocate nonces after the new physical order is fixed; nonces are never reused, including after abandoned preparations, and encode no logical name or old address. Decryption checks the expected associated data $(T,e,j)$. The privacy proof uses confidentiality; authentication does not establish malicious-server or rollback security.

\subsection{Two-Server Private Retrieval}
For an identical replicated array $C[1..m]$ and target $j$, the client samples
\[
 q_1\xleftarrow{\$}\{0,1\}^{m},\qquad q_2=q_1\oplus\mathbf e_j.
\]
Server $CS_i$ returns $a_i=\bigoplus_{t=1}^{m}q_i[t]C[t]$, where a zero selector contributes a zero row. Since $q_1\oplus q_2=\mathbf e_j$, $a_1\oplus a_2=C[j]$. Each server's selection is uniform and target-independent, including repeated reads with fresh randomness. This elementary PIR~\cite{chor1998private} uses linear-size selectors and full scans; correct recovery requires byte-identical replicas. Encryption makes the complete transcript guarantee computational.

\paragraph{A private-read example.}
Consider a four-row ciphertext table replicated identically at the two servers, with the client
requesting physical slot $j=3$. One possible random choice is
$q_1=(1,0,1,1)$, which gives $q_2=(1,0,0,1)$. The responses are
\[
 a_1=C[1]\oplus C[3]\oplus C[4],\qquad
 a_2=C[1]\oplus C[4].
\]
Their XOR is $C[3]$. Each server scans all four rows; its locally selected
subset is determined by its random selector, rather than identifying the
client's target. A later read of the same slot may instead use
$q'_1=(0,1,0,1)$ and $q'_2=(0,1,1,1)$. Both pairs recover the same ciphertext,
but either server individually receives a fresh uniform four-bit vector.
Fresh sampling does not require the vectors to be distinct: accidental
repetitions have the same distribution for repeated and different targets.
This distinction matters because privacy follows from the distribution of
the complete local transcript, not from a rule that forbids equal requests.
The server still observes that a read occurred, the table type and epoch,
and the request and response lengths. The protection of a read address
therefore leaves operation counts observable.

\subsection{Spatial Encoding and Prefix Covers}
\begin{figure}[htbp]
\centering
\includegraphics[width=0.60\linewidth]{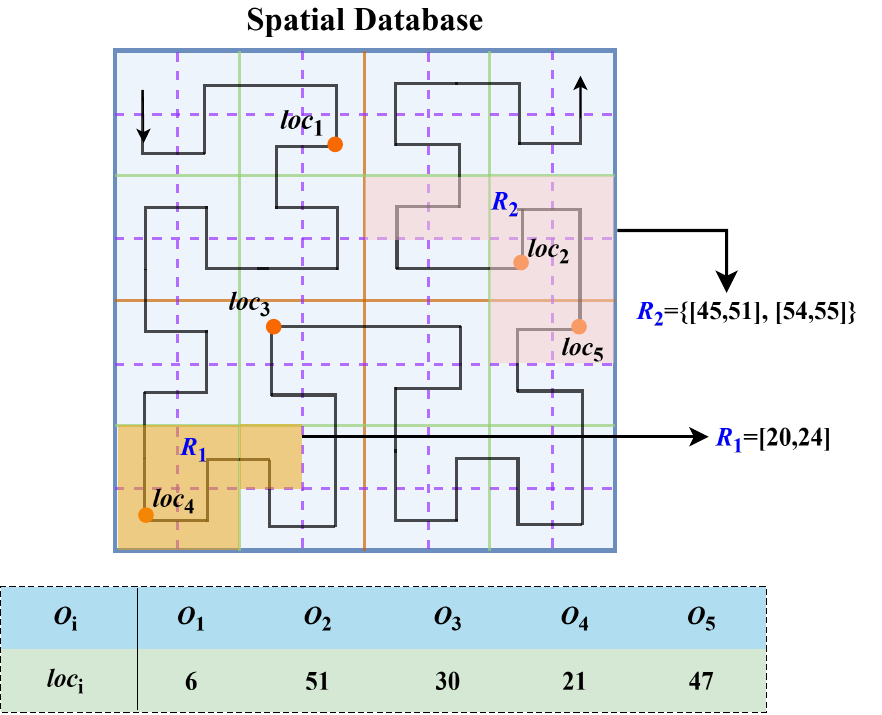}
\caption{Examples of Hilbert curves.}\label{fig:hilbert}
\end{figure}

A Hilbert curve of order $h$ assigns a $\gamma=2h$-bit value to each of $4^h$ cells. Write $X_j=H_h(\mathsf{cell}_h(loc_j))$. For a rectangle $R_q$, $\mathcal C_h(R_q)$ is its set of intersecting cells; their Hilbert values may form several intervals. Decomposing these into maximal aligned binary intervals gives an exact prefix cover $\mathcal G_h(R_q)$. For $X=a_1\cdots a_\gamma$, define
\[
 \mathcal P(X)=\{a_1\cdots a_\ell *^{\gamma-\ell}:0\leq\ell\leq\gamma\}.
\]
Then $X\in H_h(\mathcal C_h(R_q))$ iff $\mathcal P(X)\cap\mathcal G_h(R_q)\ne\emptyset$. The cover is exact for cells and has no spatial false negatives; original coordinates remove boundary false positives for arbitrary rectangles. Figures~\ref{fig:hilbert} and~\ref{fig:prefix-example} illustrate the cell ordering and exact prefix membership. The prefix universe has $2^{\gamma+1}-1$ distinct strings, including the root; BRASP stores only active predicates in its padded index.

\begin{figure}[htbp]
\centering
\includegraphics[width=0.90\linewidth]{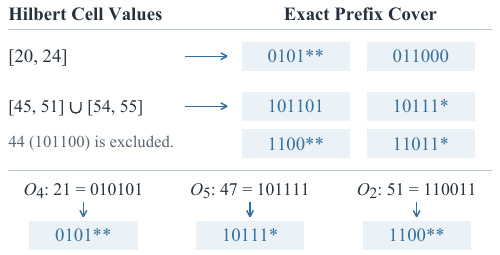}
\caption{Exact six-bit prefix covers and object membership.}
\label{fig:prefix-example}
\end{figure}

\subsection{Bitmap Semantics}
\begin{figure}[htbp]
\centering
\includegraphics[width=0.60\linewidth]{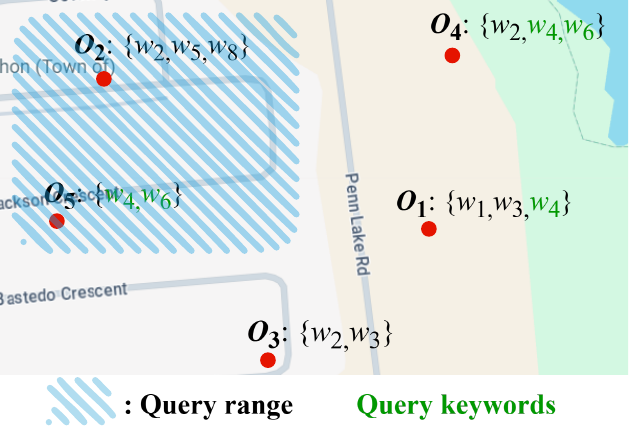}
\caption{An example of a Boolean range query.}\label{fig:BRQ}
\end{figure}

For each domain-separated name $x$, the $N$-bit bitmap $B_x$ records live-object membership: $B_{(\mathsf{keyword},w)}[j]=1$ iff $w\in W_j$, and $B_{(\mathsf{prefix},p)}[j]=1$ iff $p\in\mathcal P(X_j)$. Unoccupied positions are zero. Figure~\ref{fig:logical-indexes} shows the logical bitmaps before encryption; their columns are immutable IDs, not physical storage slots. For the predicate illustrated in Figure~\ref{fig:BRQ}, the client forms
\begin{equation}
 \begin{aligned}
 B_Q={}&V_e\land\left(\bigvee_{p\in\mathcal G_h(R_q)}B_{(\mathsf{prefix},p)}\right)\\
 &{}\land\left(\bigwedge_{w\in W_q}B_{(\mathsf{keyword},w)}\right).
 \end{aligned}
 \label{eq:bitmap-query}
\end{equation}
An empty spatial union is $0^N$; an empty keyword conjunction is $1^N$. Missing predicates contribute zero bitmaps, whereas extra padding responses are ignored. Set bits identify candidates for private retrieval and exact filtering.

\begin{figure}[htbp]
\centering
\includegraphics[width=0.90\linewidth]{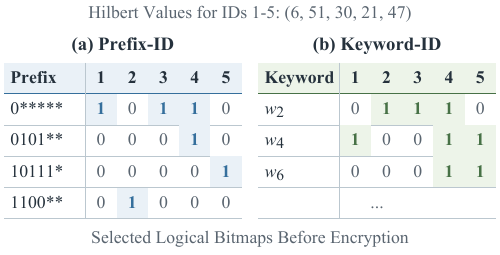}
\caption{Logical inverted indexes before encryption: (a) prefix--ID; (b) keyword--ID. Only five illustrative object IDs are shown.}
\label{fig:logical-indexes}
\end{figure}

\section{Scheme Construction}
\label{scheme_construction}

\begin{figure}[tbp]
\centering
\includegraphics[width=\linewidth]{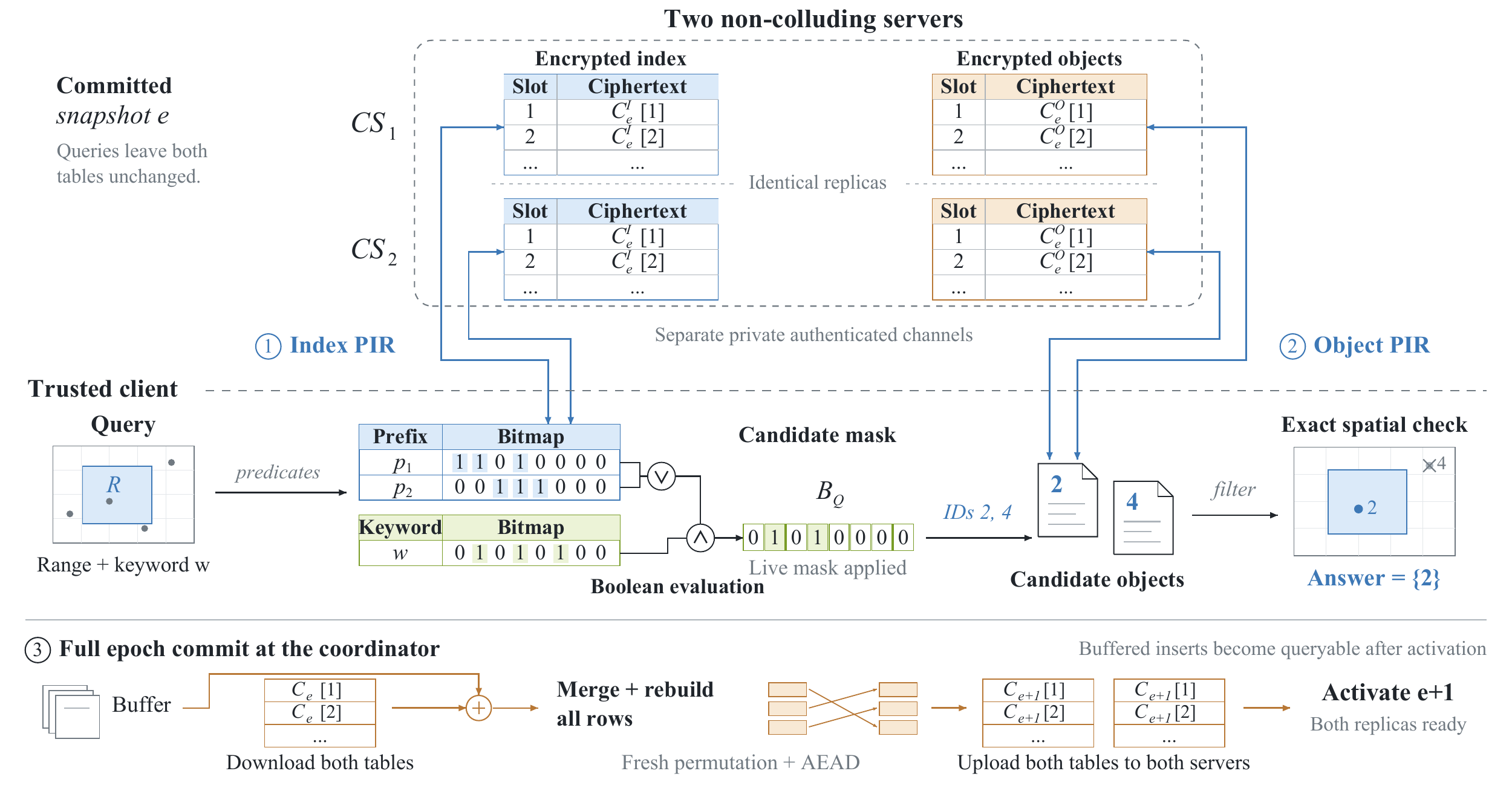}
\caption{BRASP overview: private index and object reads with local filtering; buffered insertions published through full-snapshot commits. Read counts and public scheduling remain visible.}
\label{fig:system}
\end{figure}

Figure~\ref{fig:system} separates read-only queries from insertion maintenance. Both retrieval stages use the same PIR interface, with the client alone linking predicates, bitmap coordinates, and object slots. Commits replace complete snapshots and determine insertion visibility.

\subsection{Snapshot Representation and Initialization}
\label{sec:snapshot-representation}
Object encoding is injective and contains an encrypted marker for a real or dummy record, logical ID, original coordinates, complete keywords, and payload, padded to $\ell_O$. Unused IDs and ID $0$ receive equally long dummy encodings. These records suffice to reconstruct all bitmaps. For live IDs $j$, the active predicate set is
\[
 \begin{aligned}
 \mathcal X_e={}&\{(\mathsf{keyword},w):w\in W_j,\ j\text{ live}\}\\
 &{}\cup\{(\mathsf{prefix},p):p\in\mathcal P(X_j),\ j\text{ live}\}.
 \end{aligned}
\]
Supported snapshots satisfy $|\mathcal X_e|+1\le M$; missing names resolve to $d_e$.

The $M$-row index stores active predicate bitmaps, a designated zero row, and zero padding. Private metadata $D_e,d_e,\pi_e,V_e$ resolves logical targets without exposing occupancy or names.

$\mathsf{Setup}(1^\lambda,\mathsf{pp})$ establishes the private authenticated channels, samples independent keys, initializes persistent counters, and sets $U=\emptyset$. The owner assigns initial IDs $1,\ldots,|DB_0|$. BuildSnapshot (Algorithm~\ref{alg:build-snapshot} in Appendix~\ref{app:snapshot-construction}) constructs all bitmaps, samples independent index/object permutations, and encrypts every row using Eq.~\eqref{eq:snapshot-row}, including dummy and unchanged rows. Nonces are allocated in the new physical order. Setup itself publishes no snapshot: the coordinator activates $e=0$ and its candidate directories only after both replicas acknowledge the same ciphertext bytes. Independently encrypting each replica would break PIR recovery.

\subsection{Private Reads of Index Rows and Objects}
Algorithm~\ref{alg:private-read} retrieves a row with $(m,L)=(M,L_I)$ for the index or $(R,L_O)$ for objects. The server sees the table type, epoch, and its selection vector, but no logical target. Here $(q_i[t])^L$ repeats the selection bit $L$ times. Both servers scan the complete table; the client combines their responses and decrypts with the expected epoch and slot.

\begin{algorithm}[t]
\caption{PrivateRead$(T,e,j)$}
\label{alg:private-read}
\KwIn{client target $j$; identical server arrays $C_e^T[1..m]$ with $L$-bit rows;
client key $K_T$}
\KwOut{the plaintext in slot $j$, delivered only to the client}
The client samples $q_1\xleftarrow{\$}\{0,1\}^{m}$ and sets
$q_2\leftarrow q_1\oplus\mathbf e_j$\;
Send $(T,e,q_i)$ privately to $CS_i$, for $i\in\{1,2\}$\;
\ForEach{$i\in\{1,2\}$}{
    $a_i\leftarrow 0^L$\;
    \For{$t\leftarrow 1$ \KwTo $m$}{
        $a_i\leftarrow a_i\oplus\bigl(C_e^T[t]\land(q_i[t])^L\bigr)$\;
    }
    $CS_i$ returns $(T,e,a_i)$ privately to the client\;
}
The client sets $C\leftarrow a_1\oplus a_2$ and parses $C=\nu\Vert c$\;
\Return $\mathsf{AEAD.Dec}(K_T,\nu,c,(T,e,j))$\;
\end{algorithm}

Every invocation uses fresh independent selectors, even for repeated or dummy targets. Cache hits never suppress scheduled calls. Responses and directories must belong to the same pinned epoch; authentication failure terminates the operation.

\subsection{Boolean Evaluation and Private Object Retrieval}
Let $a_Q=|W_q|+|\mathcal G_h(R_q)|$. The public policy chooses $s_Q\ge a_Q$ index reads. Actual predicates resolve through $D_e$, with missing names mapped to $d_e$; padding also reads $d_e$ but is excluded from Eq.~\eqref{eq:bitmap-query}. The resulting candidate set $J_Q$ has size $c_Q$. The client privately retrieves its object slots through $\pi_e$, pads to $v_Q\ge c_Q$ with $\pi_e(0)$, and filters genuine records by original coordinates.

\begin{algorithm}[t]
\caption{Search$(Q,e)$ with Private Object Retrieval}
\label{alg:search}
\KwIn{client state for committed epoch $e$; query $Q=(R_q,W_q)$; public padding policy}
\KwOut{matching objects and their identifiers, delivered only to the client}
Pin epoch $e$ and its directories until this query finishes\;
Compute $\mathcal G_h(R_q)$ and privately resolve each keyword/prefix through
$D_e$, using $d_e$ for an absent name\;
Set $s_Q\ge a_Q$ using $\mathsf{pol}_Q$; build a list of the $a_Q$ actual targets and $s_Q-a_Q$ targets at $d_e$;
retain their actual-predicate or padding roles only at the client\;
Randomly order the list and call $\mathsf{PrivateRead}(\mathsf I,e,j)$ for
each target $j$\;
Using only actual predicate responses, form
$B_{\mathrm{sp}}=\bigvee_{p\in\mathcal G_h(R_q)}B_{(\mathsf{prefix},p)}$ and
$B_{\mathrm{kw}}=\bigwedge_{w\in W_q}B_{(\mathsf{keyword},w)}$\;
Use $0^N$ for an empty spatial union and $1^N$ for an empty keyword conjunction\;
$B_Q\leftarrow V_e\land B_{\mathrm{sp}}\land B_{\mathrm{kw}}$\;
Let $J_Q=\{j\in\{1,\ldots,N\}:B_Q[j]=1\}$ and $c_Q=|J_Q|$\;
Set $v_Q\ge c_Q$ using $\mathsf{pol}_Q$\;
Create a list of the targets $\pi_e(j)$ for $j\in J_Q$ and $v_Q-c_Q$
targets at $\pi_e(0)$, retaining their roles privately\;
Randomly order the list and call $\mathsf{PrivateRead}(\mathsf O,e,t)$ for
every target $t$; discard padding responses\;
Return the fetched genuine objects satisfying $loc_j\in R_q$, and their IDs\;
Release epoch $e$\;
\end{algorithm}

The default policy uses $s_Q=a_Q,v_Q=c_Q$; fixed-budget mode requires public bounds $S\ge a_Q,V\ge c_Q$ for every accepted query, without secret-dependent fallback. Absent predicates and empty bitmaps never suppress scheduled reads. Counts include candidates later filtered out. Result IDs remain private, and queries perform no writes.

\paragraph{An illustrative query from predicates to records.}
Consider a separate five-object example with $N=5$, $M=64$, and $R=6$.
All five IDs are live, so $V_e=11111$. Bit strings below list logical IDs
$1,\ldots,5$ from left to right. Let their six-bit Hilbert values be
$(24,20,30,23,47)$, and suppose the cells intersecting a query rectangle have
the exact cover $\mathcal G_h(R_q)=\{0101**,01100*\}$. The client requests
keywords $W_q=\{w_a,w_b\}$. One concrete realization uses unit cells and
$R_q=[0.1,2.9]\times[6.2,7.9]$, with Hilbert values 20--25 occupying
columns 0--2 of rows 6--7 in the illustrative grid orientation.
The corresponding bitmaps are
\[
\begin{aligned}
 B_{(\mathsf{prefix},0101**)}&=01010,&
 B_{(\mathsf{prefix},01100*)}&=10000,\\
 B_{(\mathsf{keyword},w_a)}&=11011,&
 B_{(\mathsf{keyword},w_b)}&=01110.
\end{aligned}
\]
The spatial union is $11010$, and the keyword intersection is $01010$.
Equation~\eqref{eq:bitmap-query} therefore gives $B_Q=01010$ and
$J_Q=\{2,4\}$. This computation takes place at the client after decrypting
four privately retrieved index rows. For example, $D_e$ may place the two
prefix rows in physical slots 2 and 6 and the keyword rows in slots 1 and 5.
These index locations have no relationship to the bitmap columns: column 2
always denotes logical object ID 2, regardless of where any index row is
stored.

Suppose the private object directory maps $\pi_e(2)=5$ and $\pi_e(4)=1$.
The second stage privately reads object-table slots 5 and 1. Let the original
coordinates be $loc_2=(0.5,6.5)$ and $loc_4=(1.5,6.1)$. Exact filtering
returns only object 2. Object 4 is a candidate because its cell intersects
the query boundary, even though its point lies outside the rectangle.
The server receives two object-read requests in the default
mode, but neither request identifies which physical slot is recovered.
In particular, the final answer size is one whereas the observable
object-read count is two. Another query may privately read object 4 again
without exposing that shared target through its selectors.

The same example also separates absent predicates from padding. If the
public index-read budget is $S=6$, two additional reads target the zero row
$d_e$, but their responses are excluded from Boolean evaluation. If
$V=3$, one additional object read targets $\pi_e(0)$ and is discarded.
With the fixed budgets, the server observes six index reads and three object reads, including when the actual candidate set is smaller.

\paragraph{Missing predicates and empty answers.}
A missing keyword contributes a zero bitmap to the conjunction, making the candidate set empty. A missing prefix contributes zero only to the spatial union; other cover prefixes may still match. Padding rows have a different role: they make a scheduled private read but never participate in either Boolean expression. In particular, treating an index-padding row as an actual keyword would incorrectly force an empty result. The client therefore keeps the response roles locally even though the requests have the same server-visible format. With no keywords, the all-ones identity permits a purely spatial query. An empty spatial cover yields no candidates. Under fixed object-read padding, an empty candidate set still produces the scheduled dummy reads, whereas the default policy exposes a zero object-read count.

\subsection{Insertion Buffering and Epoch Rebuilding}
BufferInsert reserves an unused immutable ID and stores the record in $U$ without a server message. Commits follow the public policy, such as a fixed batch size or predetermined update slots, independently of keyword contents, range membership, or new-name counts. To hide sparse write destinations, Algorithm~\ref{alg:commit} rebuilds both complete tables.

\begin{algorithm}[t]
\caption{CommitEpoch$(e,U)$}
\label{alg:commit}
\KwIn{committed snapshot $e$; trusted pending buffer $U$ of $b$ insertions}
\KwOut{committed snapshot $e+1$ and refreshed private client state}
Stop admitting queries and finish all queries pinned to epoch $e$\;
Download all $M$ index rows and $R$ object rows from $CS_1$ in fixed physical
order; verify and decrypt them locally\;
Recover $DB_e$ and add $U$ using its reserved logical identifiers\;
Build the entire next snapshot with
$\mathsf{BuildSnapshot}(DB_e\cup U,e+1)$\;
Completely prepare the two new ciphertext arrays and private directories\;
Upload the identical complete arrays, in physical order, to staging areas at
$CS_1$ and $CS_2$\;
Wait for both servers to acknowledge that version $e+1$ is ready\;
Persist the new directories, live mask, epoch, and consumed nonce counters\;
Activate $e+1$, clear $U$, and resume queries using only the new version\;
\end{algorithm}

Rebuilding incorporates the batch and newly active names, then reencrypts all real, unchanged, and dummy rows under fresh independent permutations. Full download precedes preparation and upload, preventing per-record transfer pairing; no affected-row list or old-to-new mapping is transmitted. Both replicas retain the old epoch during staging. Reads resume after both acknowledgments and activation of matching metadata. Atomic switching relies on the trusted serialized coordinator and provides neither distributed fault tolerance nor rollback protection.

A batch of one makes each insertion visible after its own rebuild; larger batches amortize maintenance and delay visibility. Section~\ref{sec:phase-costs} accounts for workspace and the alternative of an authoritative local database.

Buffering leaves $V_e$ unchanged. At commit, genuine records reconstruct $DB_e$, and decoded index rows permit a consistency check. No server-to-server messages are required.

\paragraph{Visibility across a successful commit.}
The insertion buffer and the committed snapshot have different visibility
rules. Reserving an unused logical ID places the new record in $U$ without
changing $V_e$, $D_e$, or either encrypted replica. A query pinned to epoch
$e$ consequently evaluates exactly $DB_e$, even if the buffer contains
objects that satisfy its predicate. This delay is a property of the
published database version, rather than a failed query or an approximate
answer.

At the public commit boundary, the coordinator stops admitting queries and
waits for all reads of epoch $e$ to finish. It then downloads and decrypts
the complete old tables, constructs $DB_e\cup U$, and prepares the entire
next snapshot locally. All rows receive fresh encryption, including rows
whose logical contents are unchanged and rows reserved for padding. Fresh
permutations change physical destinations while preserving each existing
object's logical ID and bitmap column. Preparing the complete arrays before
upload also prevents the transfer sequence from pairing an old row with
its replacement.

Both servers receive the same new ciphertext bytes in their staging areas.
Only after both acknowledgments does the coordinator persist the matching
directories and live mask, activate epoch $e+1$, clear the buffer, and resume
queries. A query never combines an index bitmap from epoch $e$ with an
object record from epoch $e+1$. The protocol's successful execution thus
has a single visibility boundary for the entire batch. This ordering uses
the trusted serialized coordinator assumed by the model; the analysis of
interrupted installation and recovery remains separate from the completed
commit guarantee.

\subsection{Resource Costs by Phase}
\label{sec:phase-costs}\label{theoretical_analysis}

All expressions count payload bits, including nonces, tags, and row padding but excluding headers and transport framing. One snapshot occupies
\begin{equation}
 D=ML_I+RL_O.\label{eq:snapshot-size}
\end{equation}
Identical active replicas use $2D$ bits, or up to $4D$ during staging. Dummy rows occupy space and participate in every scan and rebuild.

\paragraph{Private reads and queries.}
Reading an $L$-bit row from $m$ rows uploads two $m$-bit selectors and downloads two $L$-bit answers: $2(m+L)$ bits. Each server performs $\Theta(mL)$ bit operations, or $\Theta(m\lceil L/w\rceil)$ word operations for word size $w$, plus selector handling. Thus
\begin{align}
 \operatorname{Comm}_Q&=2s_Q(M+L_I)+2v_Q(R+L_O),\label{eq:query-cost}\\
 \operatorname{Work}_Q^{(i)}&=\Theta(s_QML_I+v_QRL_O).\label{eq:query-work}
\end{align}
The client generates $s_QM+v_QR$ random bits and complementary selectors, combines $s_Q+v_Q$ answer pairs, decrypts the recovered rows, and evaluates $O(a_QN)$ bitmap bits, plus range-cover construction and coordinate checks for $c_Q$ retrieved candidates. Streaming needs bitmap accumulators, one pair of selectors/responses, and output records. Padding substitutes public budgets for $s_Q=a_Q,v_Q=c_Q$; insufficient budgets are outside the supported workload. Batching may improve memory reuse but not yield sublinear search.

The two terms in Eq.~\eqref{eq:query-cost} expose different sources of query cost. Spatial decomposition and keyword count determine the actual index-read count, while the intersection of the recovered bitmaps determines how many object records require retrieval. For fixed capacities, each additional candidate adds $2(R+L_O)$ payload bits in the unpadded policy, even if exact coordinate filtering later excludes it. Consequently, the final answer cardinality alone does not determine communication. Finer spatial cells can reduce boundary candidates but also change the number of cover prefixes and active index predicates. Their net effect depends on both stages.

For the same query, fixed read budgets $(S,V)$ add the following payload relative to unpadded execution:
\begin{equation}
 \Delta\operatorname{Comm}_Q
 =2(S-a_Q)(M+L_I)+2(V-c_Q)(R+L_O).
 \label{eq:padding-cost}
\end{equation}
This expression charges dummy reads at their full private-read cost. Keeping $S$ and $V$ constant also fixes the direct scan work for accepted queries at fixed capacities. It does not hide whether a query occurs or which public epoch serves it. We choose the budgets as part of the public workload policy rather than adapting them after observing a secret-dependent overflow.

\paragraph{Construction and commits.}
Initialization encrypts $M+R$ rows and uploads $2D$ bits; buffering communicates nothing. A baseline commit downloads one snapshot, decrypts all $M+R$ rows, and encrypts/uploads both new replicas:
\begin{equation}
 \operatorname{Comm}_{\mathrm{Commit}}=D+2D=3D.\label{eq:commit-cost}
\end{equation}
Logical rebuilding initializes $MN$ bitmap positions and processes $O(MN+J)$ positions, where
\[
 J=\sum_{O_j\in DB_e\cup U}(|W_j|+\gamma+1),\qquad\gamma=2h.
\]
The prefix term counts incidences including the root, not distinct rows. Encoding, directories, permutations, and serialization add work. A batch of $b\ge1$ shares $3D/b$ bits per insertion but retains rebuilding, delayed visibility, and paused queries. An authoritative owner database trades persistent storage for a $2D$ commit. Public capacity expansion transfers the old size plus twice the new size.

Batching changes how often a complete snapshot is replaced. For $B$ insertions divided into $k$ nonempty commits within the same fixed capacities, total maintenance payload is $3kD$, and its average per insertion is $3kD/B$. For example, publishing 100 insertions one at a time transfers $300D$, whereas one commit containing all 100 transfers $3D$. This is an accounting example for the stated protocol, not an experimental timing result. Both schedules eventually publish the same number of insertions, but objects in the larger batch wait for its commit to complete. They also incur one full logical rebuild at each commit, so payload amortization alone does not specify client computation or service downtime.

\paragraph{Client storage.}
Mappings and the live mask require $O(M\log M+R\log R+N)$ bits, besides predicate names, keys, counters, and $U$. Direct rebuilding uses $O(MN+R\ell_O)$ plaintext bits and $O(D)$ ciphertext staging space. External-memory variants must count local I/O and preserve full-download-before-upload ordering.

\paragraph{Design tradeoffs.}\label{app:comparison}
DPF-based PIR~\cite{gilboa2014distributed} can shorten requests but requires a different construction and cost analysis; short requests do not imply sublinear scans. Path ORAM~\cite{stefanov2013path} offers private block reads/writes with a position map and stash, potentially avoiding full replacement, but must charge every bitmap/object block. Range DSSE~\cite{wangchow2022range} has different capabilities. Comparisons must align semantics, leakage, client state, and insertion visibility.

\section{Security Analysis}
\label{sec:security}

We analyze private index/object reads and complete insertion commits jointly in either server's retained history. All guarantees use the trust, channel, and scheduling assumptions of Section~\ref{sec:problem}.

\subsection{Histories and Explicit Leakage}
A supported history $\mathcal H$ consists of setup, queries, local insertions, and completed commits, starting with fresh keys and zero counters. Snapshot construction and installation are uninterrupted: crashes, abandoned preparations, and retries are outside the theorem. Each construction encrypts exactly one row per new physical slot; records, databases, queries, and nonce consumption satisfy their public bounds. Queries use one committed epoch. The public schedule fixes operation boundaries, transmission order, and the designated download server, independently of affected contents.

Setup leakage covers local Setup, BuildSnapshot$(DB_0,0)$, and installation. For a query with $s_Q$ index and $v_Q$ object reads, and a commit of $b$ objects, define
\begin{align}
 L_{\mathsf{Setup}}&=(\lambda,\mathsf{pp},R,L_I,L_O),\label{eq:leakage-setup}\\
 L_{\mathsf{Query}}(\mathcal H,Q)&=(e,s_Q,v_Q),\label{eq:leakage-query}\\
 L_{\mathsf{Commit}}(\mathcal H,b)&=(e,e+1,b,M,R,L_I,L_O).\label{eq:leakage-commit}
\end{align}
Local buffering sends no message. The history leakage $\mathcal L^{(i)}(\mathcal H)$ is server identity $i$, the public schedule, and these ordered releases. They determine fixed scan/download addresses, transfer lengths, acknowledgments, and headers. Capacities are not occupied-row counts; expansion is outside this bounded-capacity analysis.

The leakage excludes logical addresses, modified predicates, result identities/intersections, and private directories. In unpadded mode $s_Q=a_Q,v_Q=c_Q$; fixed budgets replace these counts. Volumes and timing can still support inference, including changed candidate counts after insertion. Pattern privacy is therefore relative to $\mathcal L^{(i)}$.

\paragraph{Interpreting the privacy boundary.}
An address relation and a volume relation are different observations. Two queries with the same scheduled read counts may access identical predicates, disjoint predicates, or overlapping candidate sets. The server's selections do not additionally distinguish these cases under the theorem's equal-leakage comparison. Queries with different counts remain distinguishable through those counts. In particular, the unpadded object-read count describes candidates before exact filtering, rather than the number of objects returned to the client. An observer may combine this count with auxiliary knowledge; the guarantee does not erase information already contained in the declared leakage.

The same distinction applies across commits. Full replacement hides which old rows gained membership bits or which padded slots now represent new predicates. The epoch, batch size, and transfer schedule remain public. If a subsequent query has a larger candidate count, that later count is still released in unpadded mode. The forward-privacy statement therefore compares batches through their commit transcript and extends to later operations only when their additional leakage also agrees. This makes the meaning of protection consistent for both the sparse-write example in Figure~\ref{fig:motivation}(c) and the complete histories analyzed below.

\subsection{Correctness}
\begin{theorem}[Correctness]
\label{thm:correctness}
For a supported history with honest execution, correct AEAD encryption and
decryption, and identical server replicas, every query returns exactly the
objects in the committed database satisfying its original spatial range and
conjunctive keyword predicate. Every completed commit preserves this property
for the database extended by its buffered insertion batch.
\end{theorem}

\begin{proof}
XORing the PIR answers recovers the requested ciphertext. Correct decryption, the exact cell-prefix cover, and Eq.~\eqref{eq:bitmap-query} recover all candidates; coordinate filtering removes boundary false positives. Dummy padding does not participate in evaluation. A commit rebuilds these invariants for the extended database and installs identical replicas before activating the new epoch. Appendix~\ref{app:security-proofs} supplies the full invariants, including missing predicates and empty keyword sets.
\end{proof}

\subsection{Privacy of the Complete Server View}
The corrupted server's view includes every retained snapshot, its own PIR selections/responses, and all transfers involving it. The ideal functionality maintains the logical database and buffer, evaluates queries adaptively, and delivers plaintext answers only to the client. The simulator receives $\mathcal L^{(i)}$, not logical addresses or answers.

\begin{theorem}[Single-server transcript privacy]
\label{thm:transcript}
Assume the AEAD scheme provides IND-CPA confidentiality for multiple
messages with distinct nonces. Also assume independent secret keys for
the two table types, fresh independent uniform PIR selections, and the
channel and scheduling model above. For
either $i\in\{1,2\}$ and every PPT adversary $A$, there is a PPT simulator
$S_i$ such that, for supported polynomial-length adaptive histories,
\[
\left|\Pr[\mathsf{Real}_{A}^{(i)}(\lambda)=1]
-\Pr[\mathsf{Ideal}_{A,S_i}^{(i)}(\lambda)=1]\right|
\leq \operatorname{negl}(\lambda).
\]
\end{theorem}

\noindent\textit{Proof idea.}
Replace every newly encrypted payload by an equal-length zero string while preserving nonces, associated data, snapshot copies, and ciphertext reuse. Multi-message AEAD confidentiality hides this change. Each fresh selection seen by one server is uniform regardless of an adaptively chosen target; its response is determined by that selection and the stored table. A simulator can thus generate dummy encrypted snapshots, selections, responses, and complete transfers from leakage. The ideal functionality still evaluates the real database; the hybrid does not ask clients to derive answers from zero payloads. Appendix~\ref{app:security-proofs} gives the game sequence.

\begin{proposition}[Epoch correspondence hiding]
\label{prop:rebuild}
Under the assumptions of Theorem~\ref{thm:transcript}, a completed rebuild
reveals no additional old-to-new logical row correspondence beyond
$\mathcal{L}^{(i)}$. In particular, two supported histories with the same
leakage and different hidden correspondences produce computationally
indistinguishable server views.
\end{proposition}

Fixed-order full transfers, secret fresh permutations, and fresh encryption of unchanged/dummy rows expose no per-row correspondence. Public nonces refer only to the new physical ordering.

\begin{corollary}[Forward privacy of committed insertions]
\label{cor:forward}
Under the assumptions of Theorem~\ref{thm:transcript}, a commit leaks no
additional relation between its inserted objects and earlier queries beyond
$\mathcal{L}^{(i)}$. Specifically, two insertion batches of the same public
size, extending the same preceding logical history within the supported
capacities, yield indistinguishable server views through the commit whenever
their public commit metadata agree.
\end{corollary}

This guarantee concerns the commit transcript: complete replacement reveals no affected-row list, and earlier selections hide query targets. Subsequent queries retain Eq.~\eqref{eq:leakage-query}, including candidate-count changes. The proof covers insertion-only histories under single-server corruption.

\section{Experiments}
\label{experiment}

\subsection{Experimental Setup and Measurement Scope}
\label{app:exp_settings}
We evaluate BRASP and the $\mathrm{VPBRQ}_{SupL}$~\cite{tong2024beyond} and PPSKS~\cite{zhang2022efficient} implementations on Yelp business data\footnote{\url{https://business.yelp.com/dataset}}. Comparison sweeps use 20,000--100,000 objects, varying object count $n$, vocabulary size $m$, or query-keyword count $|W_q|$; storage measurements extend to 200,000 objects. Insertion workloads contain $w_o$ objects. These variables differ from padded capacities $N,M$ and commit batch size $b$. The platform, parameter grid, and accounting conventions are specified below.

The figures report recorded construction, request-generation, search, and insertion-workload costs. Search measurements do not separate index PIR, Boolean evaluation, object PIR, and filtering; $w_o$ does not identify a complete commit. Section~\ref{sec:phase-costs} provides complete-protocol accounting. Baseline functionality also differs: $\mathrm{VPBRQ}_{SupL}$ verifies results, whereas PPSKS targets keyword similarity. The plotted counters do not establish matched verification or object-delivery boundaries, so comparisons concern the evaluated operation measurements.

\paragraph{Platform and workload.}
Experiments use Python~3.12, two Ubuntu~22.04 LTS servers linked by 10\,GbE, and a separate client. Each server has an Intel Xeon Gold 6330 (28 cores), 256\,GB RAM, and 2\,TB NVMe storage; the client has an Intel Core i7-class CPU and 32\,GB RAM. Yelp businesses retain coordinates, keywords, and payloads as described in Section~\ref{sec:problem}.

The grid uses $n\in\{2,\allowbreak4,\allowbreak6,\allowbreak8,\allowbreak10,\allowbreak15,\allowbreak20\}\times10^4$ objects, $m\in\{100,\allowbreak200,\allowbreak300,\allowbreak400,\allowbreak500\}$ keywords, $|W_q|\in\{2,4,6,8,10\}$ query keywords, and $w_o\in\{10,100,1000\}$ inserted objects. Vocabulary $m$ differs from index capacity $M$. Only a single-batch workload has $b=w_o$; otherwise batch sizes and commit count must be specified. Protocol configuration additionally includes $h,N,M,R,L_I,L_O$, active predicates, and padding; query descriptors include $s_Q,c_Q,v_Q$ and final match count.

\paragraph{Measurements and accounting.}
Construction/search/insertion plots use seconds; request generation and Table~\ref{tab:search_breakdown_appendix} use milliseconds. Traffic retains each display's stated units. Each directed transfer is counted once. Complete query latency follows the communication/client critical path, without adding parallel server runtimes or charging commits to each query. It includes both PIR phases, bitmap evaluation, object decryption, and exact filtering.

For serialized index/object replica sizes $S_I,S_O$ in bytes, $D=8(S_I+S_O)$. Publication transfers $2(S_I+S_O)$ bytes and a fixed-capacity commit $3(S_I+S_O)$, with index-only lower bounds $2S_I,3S_I$. The amortization denominator is the actual batch size $b$. These payloads include cryptographic padding but exclude transport/storage-allocation overhead. One byte is eight bits, MB is $10^6$ bytes, and MiB is $2^{20}$ bytes; matched capacities and row formats, not object count alone, are required for comparison.

\subsection{Results and Discussion}
\begin{figure*}[t]
    \centering
    \begin{minipage}[t]{0.32\linewidth}
      \vspace{0pt}\centering
      \begin{subfigure}[t]{\linewidth}
        \centering
        \includegraphics[width=\linewidth]{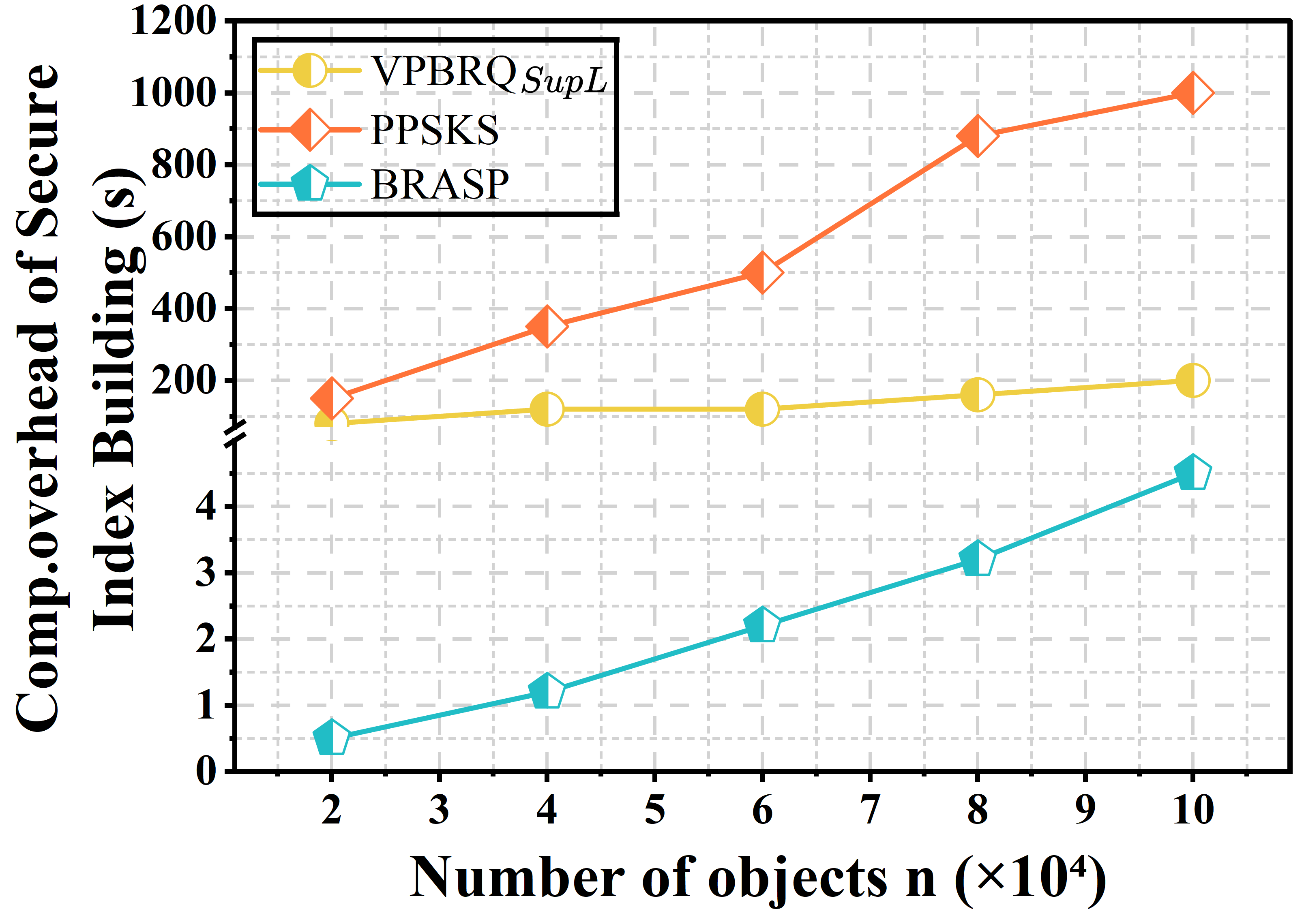}
        \caption{}
        \label{sub1}
      \end{subfigure}

      \begin{subfigure}[t]{\linewidth}
        \centering
        \includegraphics[width=\linewidth]{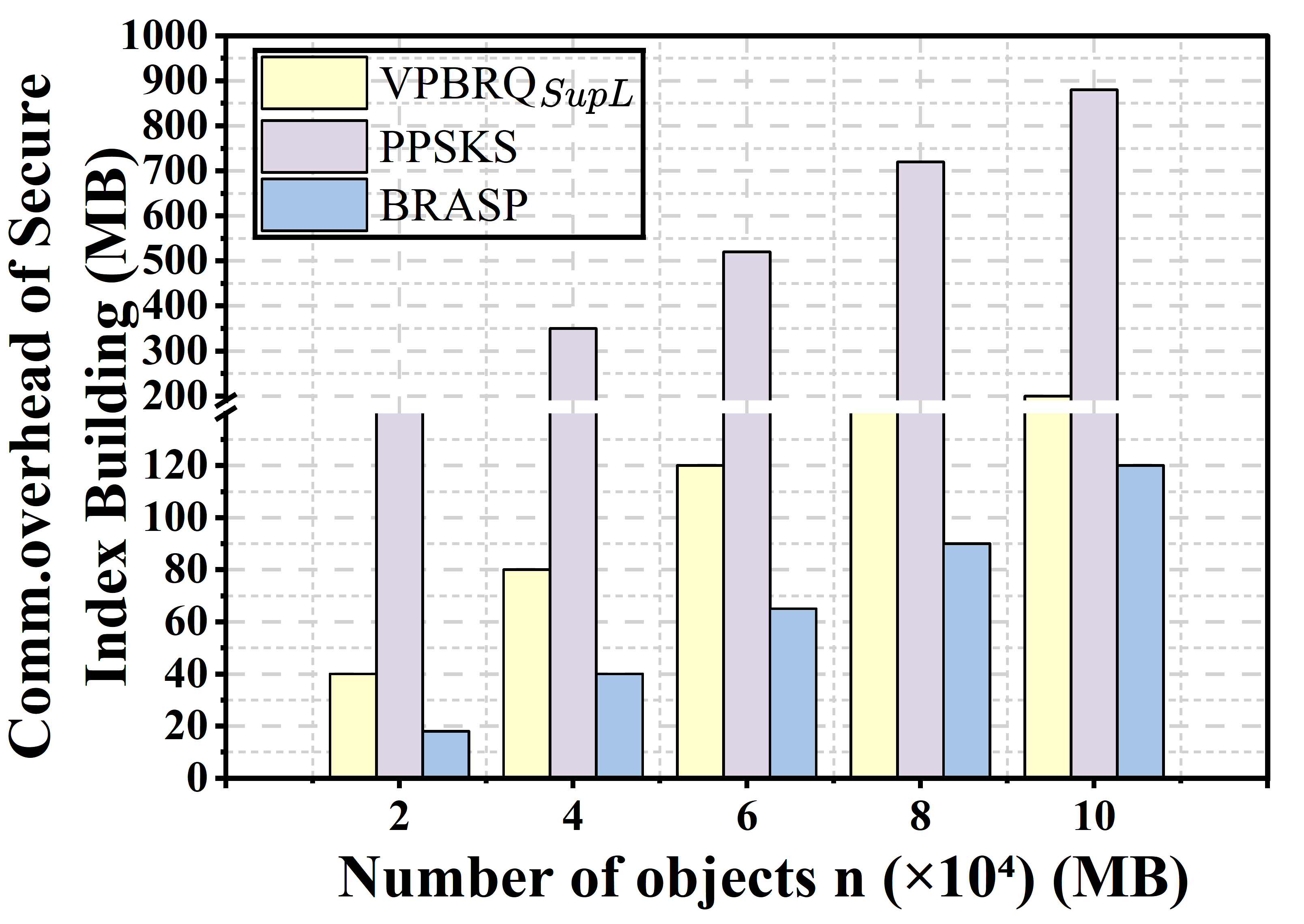}
        \caption{}
        \label{sub2}
      \end{subfigure}
    \end{minipage}\hspace{0.015\linewidth}%
    \begin{minipage}[t]{0.32\linewidth}
      \vspace{0pt}\centering
      \begin{subfigure}[t]{\linewidth}
        \centering
        \includegraphics[width=\linewidth]{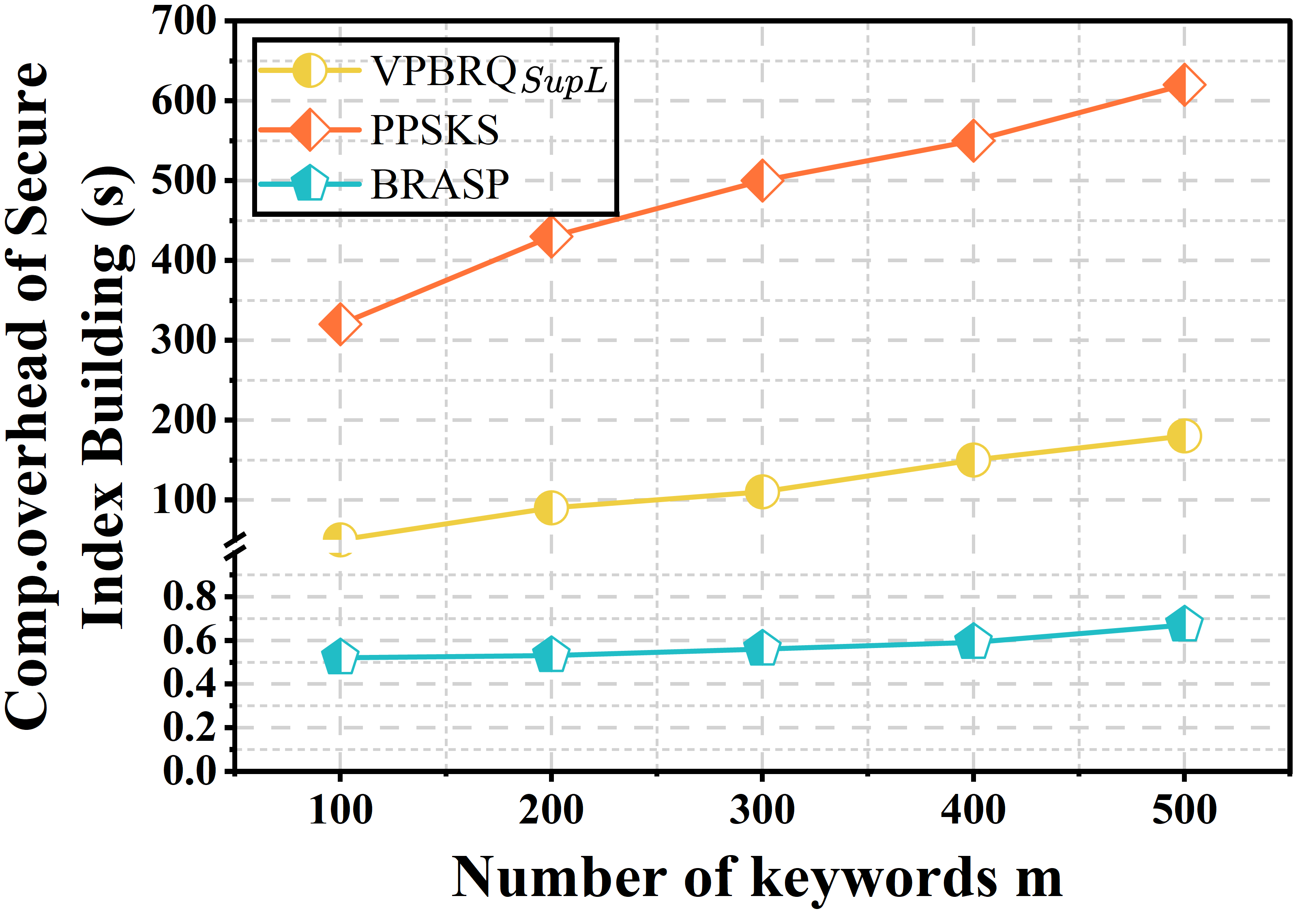}
        \caption{}
        \label{sub3}
      \end{subfigure}

      \begin{subfigure}[t]{\linewidth}
        \centering
        \includegraphics[width=\linewidth]{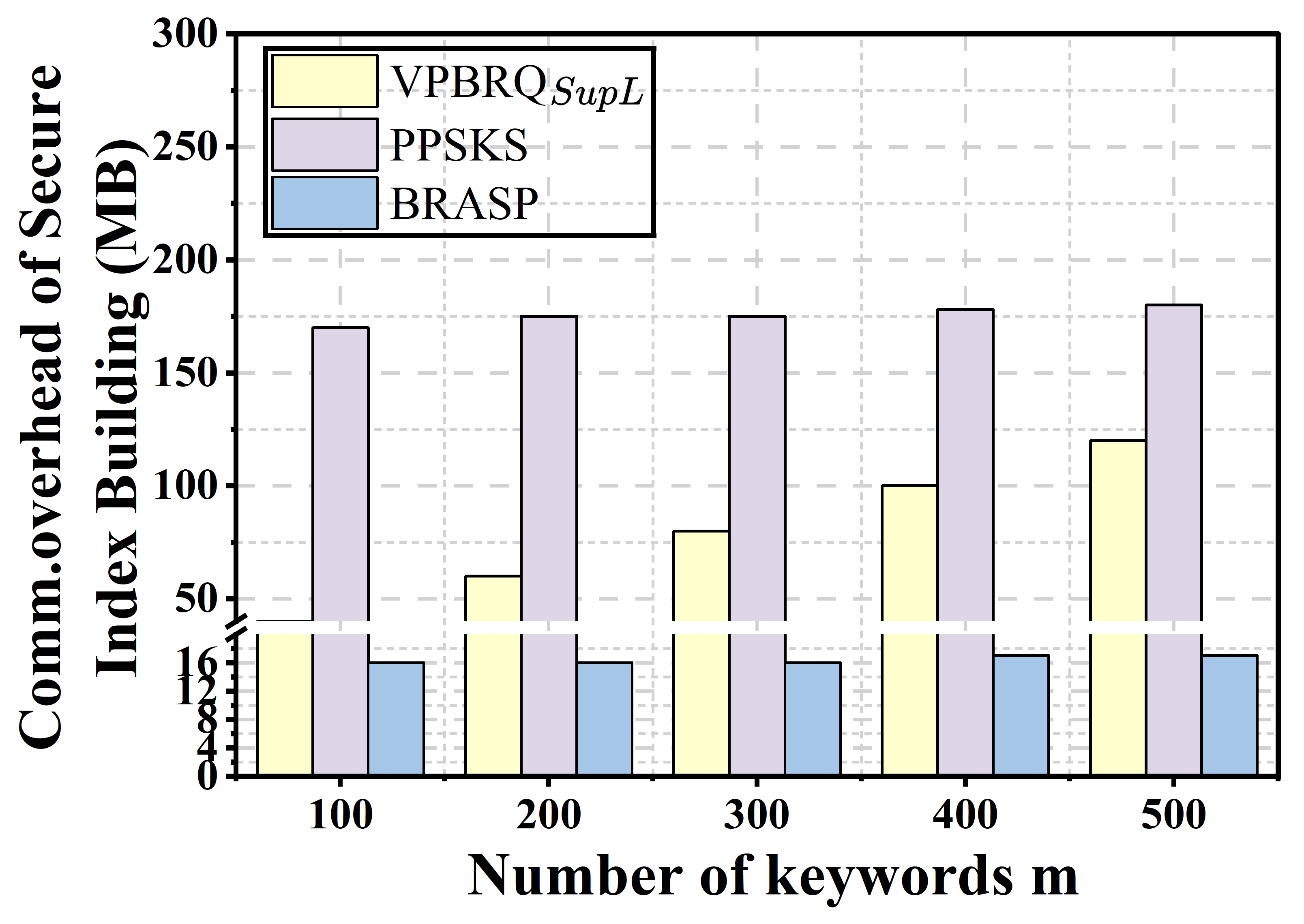}
        \caption{}
        \label{sub4}
      \end{subfigure}
    \end{minipage}\hspace{0.015\linewidth}%
    \begin{minipage}[t]{0.32\linewidth}
      \vspace{0pt}\centering
      \begin{subfigure}[t]{\linewidth}
        \centering
        \includegraphics[width=\linewidth]{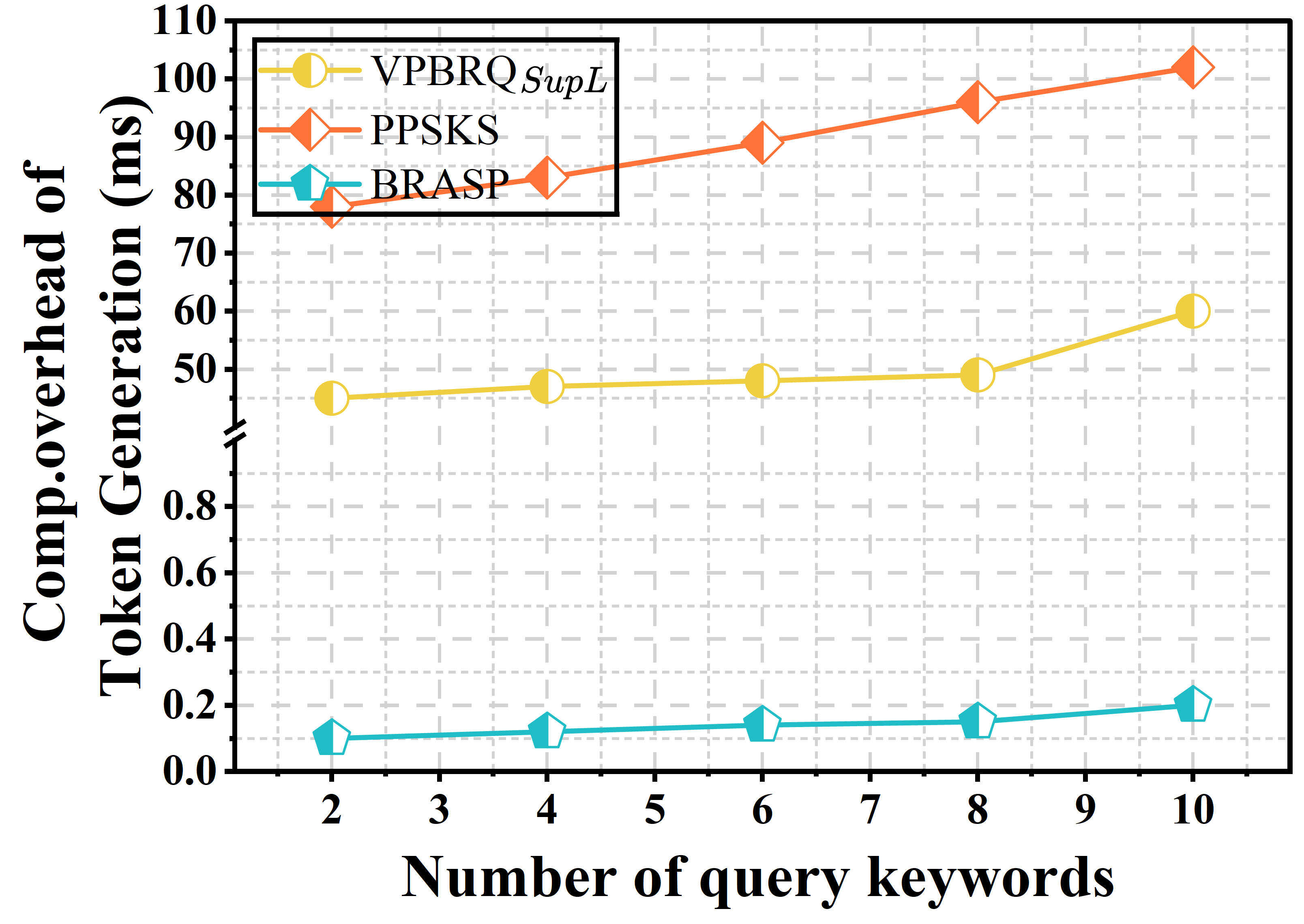}
        \caption{}
        \label{trap1}
      \end{subfigure}

      \begin{subfigure}[t]{\linewidth}
        \centering
        \includegraphics[width=\linewidth]{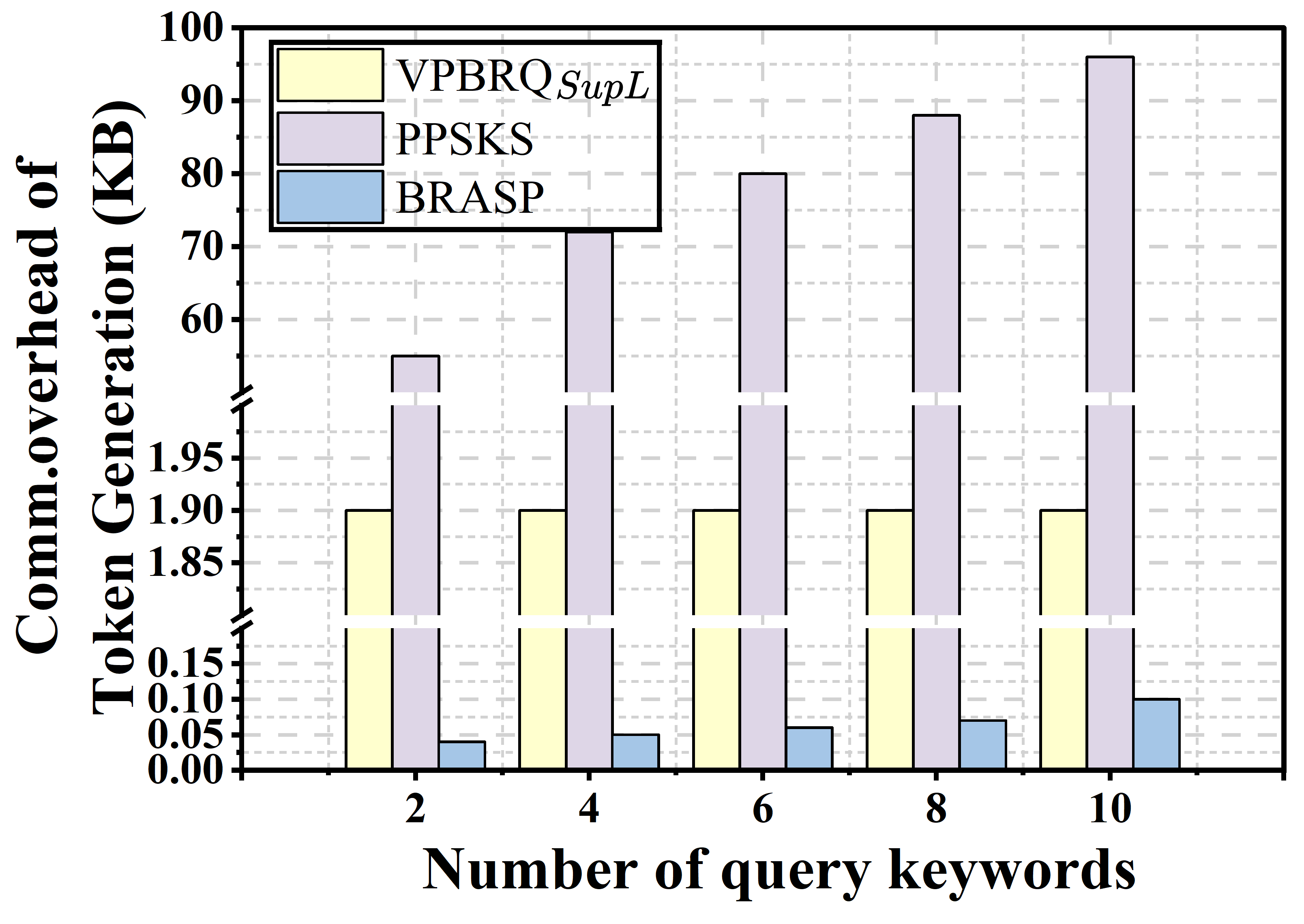}
        \caption{}
        \label{trap2}
      \end{subfigure}
    \end{minipage}

    \caption{Construction and request-generation costs:
    (a, b) construction versus object count $n$;
    (c, d) construction versus vocabulary size $m$;
    (e, f) request generation versus query-keyword count.
    The top and bottom rows report time and communication, respectively.
    Panels (e, f) measure only the instrumented request component.}
    \label{fig:index}
\end{figure*}

\paragraph{\textbf{Construction.}}
BRASP has lower construction time and communication than both evaluated baselines throughout Fig.~\ref{fig:index}(\subref{sub1}--\subref{sub4}). Object-count growth has a stronger effect than vocabulary growth from 100 to 500. Every predicate bitmap has $N$ bits, while vocabulary growth changes the active predicate count; stored size additionally depends on padded capacity $M$. These measured trends describe the tested configurations, whereas the full initial-publication payload is $2D$ for both tables and replicas.

\paragraph{\textbf{Request generation.}}
In Fig.~\ref{fig:index}(\subref{trap1}, \subref{trap2}), BRASP has the lowest recorded generation time, increasing gradually with query-keyword count. Its request-component traffic also grows, while $\mathrm{VPBRQ}_{SupL}$ remains nearly constant. The plot label ``Token Generation'' denotes this component microbenchmark. Complete queries upload $2(s_QM+v_QR)$ selector bits and receive two ciphertext responses per read; the microbenchmark does not measure this complete traffic.

\begin{figure*}[t]
    \centering
    \begin{minipage}[t]{0.32\linewidth}
        \centering
        \includegraphics[width=\linewidth]{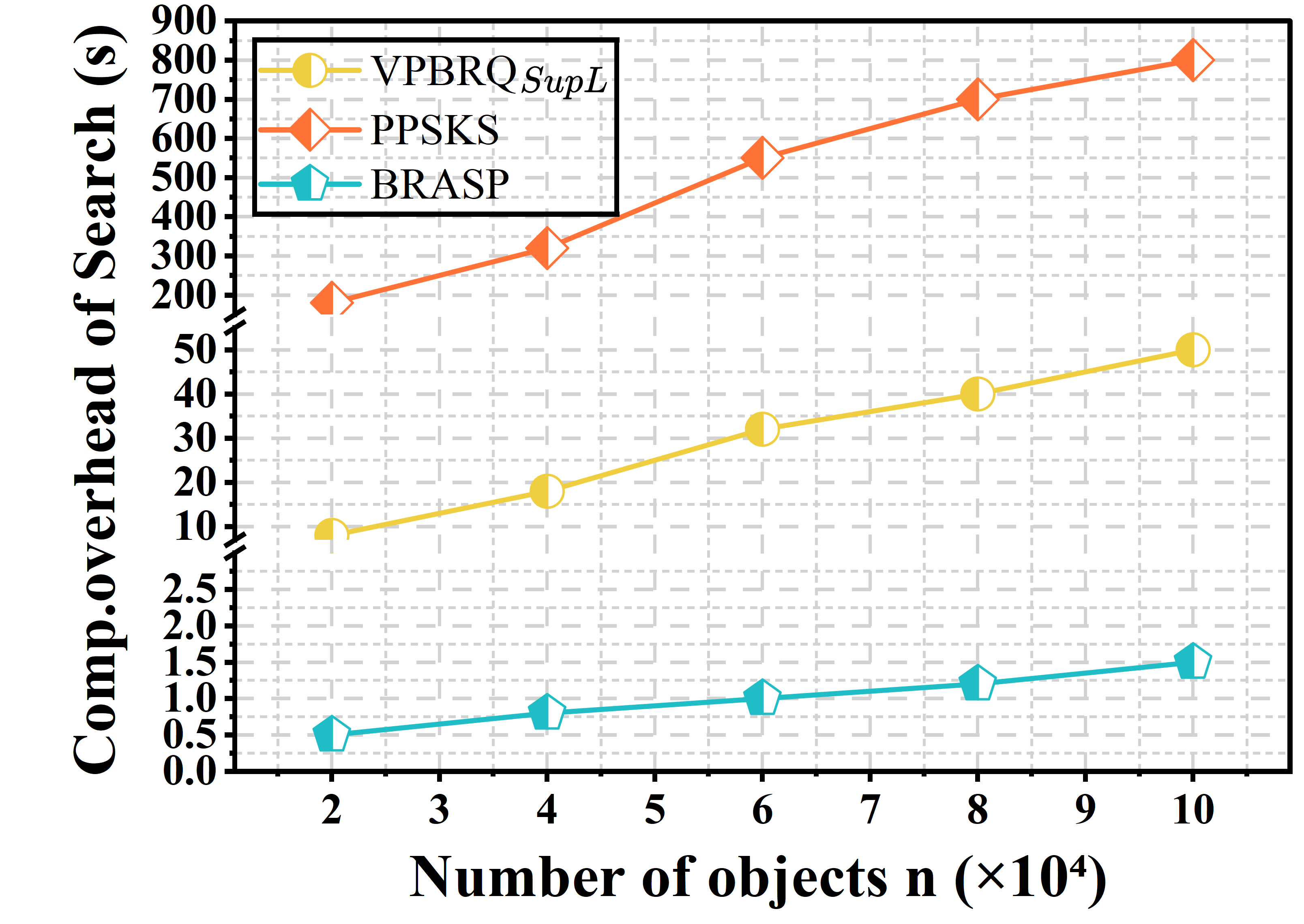}
        \subcaption{}
        \label{su1}
    \end{minipage}\hspace{0.015\linewidth}%
    \begin{minipage}[t]{0.32\linewidth}
        \centering
        \includegraphics[width=\linewidth]{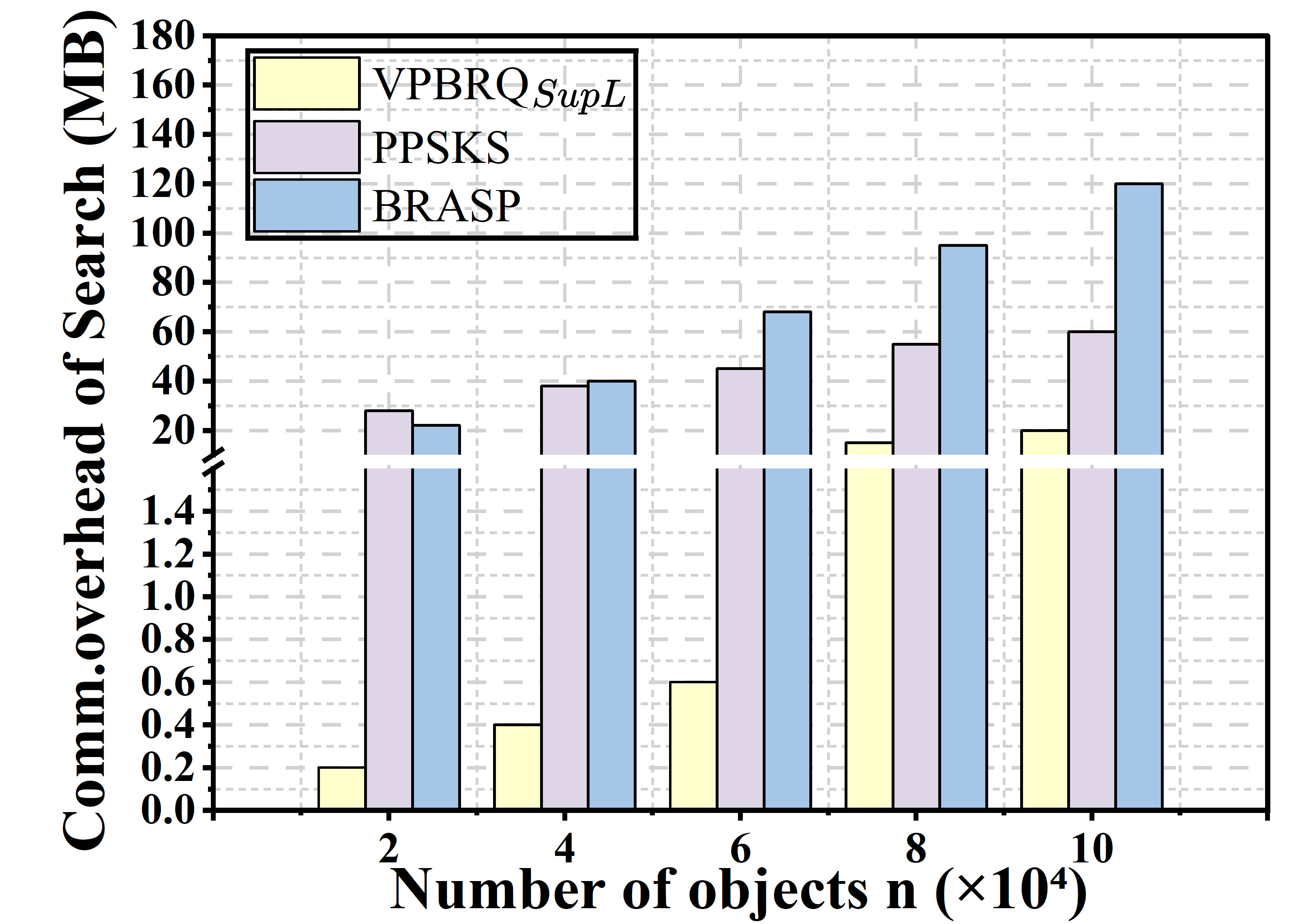}
        \subcaption{}
        \label{su2}
    \end{minipage}\hspace{0.015\linewidth}%
    \begin{minipage}[t]{0.32\linewidth}
        \centering
        \includegraphics[width=\linewidth]{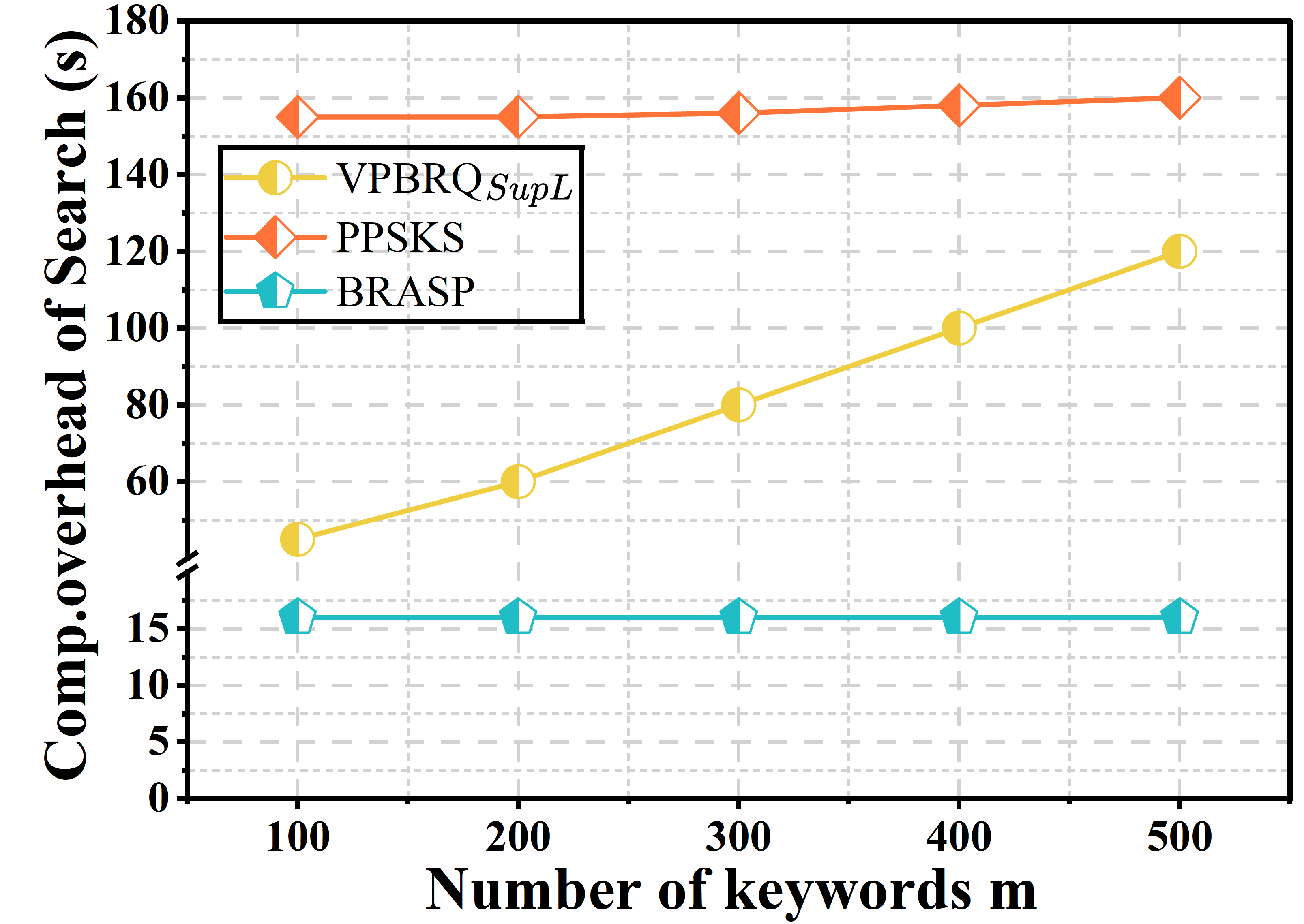}
        \subcaption{}
        \label{su3}
    \end{minipage}

    \begin{minipage}[t]{0.32\linewidth}
        \centering
        \includegraphics[width=\linewidth]{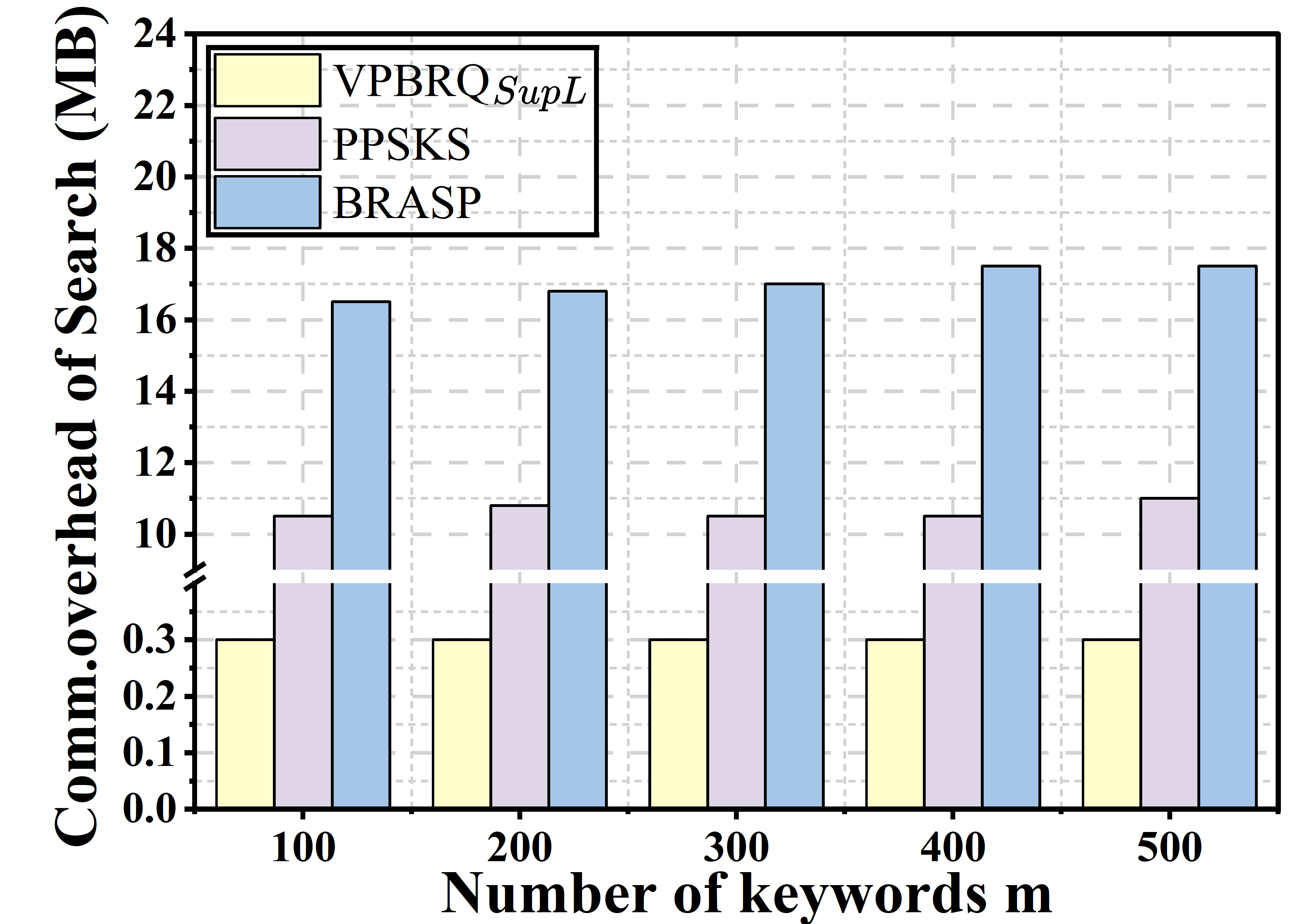}
        \subcaption{}
        \label{su4}
    \end{minipage}\hspace{0.015\linewidth}%
    \begin{minipage}[t]{0.32\linewidth}
        \centering
        \includegraphics[width=\linewidth]{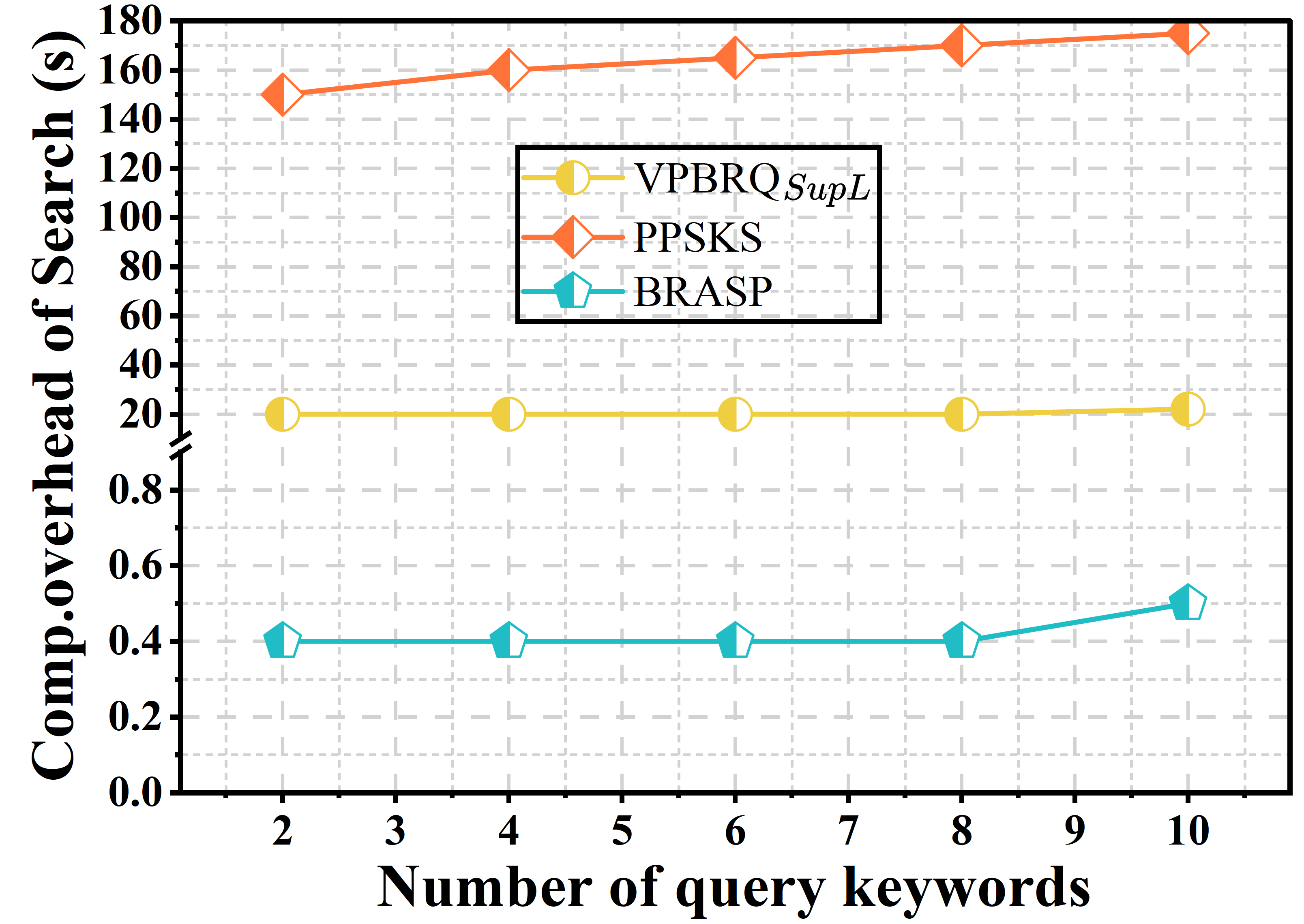}
        \subcaption{}
        \label{su5}
    \end{minipage}\hspace{0.015\linewidth}%
    \begin{minipage}[t]{0.32\linewidth}
        \centering
        \includegraphics[width=\linewidth]{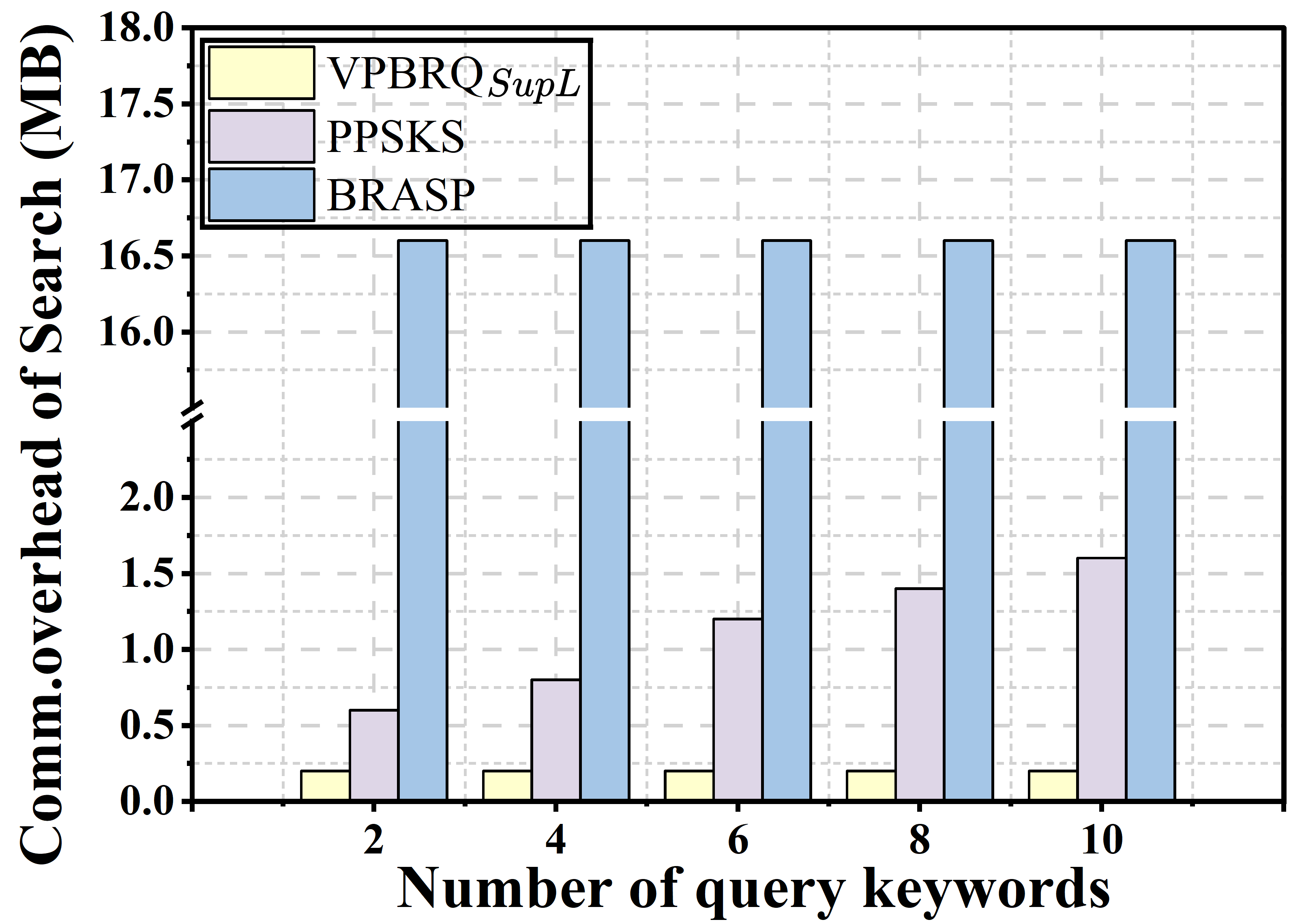}
        \subcaption{}
        \label{su6}
    \end{minipage}%
    \caption{Query computation and communication under three workload sweeps:
    (a, b) object count; (c, d) keyword-vocabulary size; (e, f) query-keyword count.}
    \label{Search}
\end{figure*}

\begin{table}[!htbp]
\caption{Online-search latency and communication for three workload configurations.}
\label{tab:search_breakdown_appendix}
\centering
\resizebox{\linewidth}{!}{
\begin{tabular}{c c c c c c}
\toprule
Setting & $n$ & $m$ & $|W_q|$ &
Online Search (ms) & Online Search (MB) \\
\midrule
S1 & $2\times 10^4$  & 100 & 2  & 61  & 0.8 \\
S2 & $10\times 10^4$ & 300 & 6  & 129 & 3.1 \\
S3 & $20\times 10^4$ & 500 & 10 & 278 & 6.5 \\
\bottomrule
\end{tabular}
}
\end{table}

\paragraph{\textbf{Search time and communication.}}
BRASP has the lowest recorded search time in all three sweeps of Fig.~\ref{Search}. Its communication exceeds $\mathrm{VPBRQ}_{SupL}$ at every tested object count and PPSKS from $4\times10^4$ objects onward; both baselines also use less communication in the vocabulary and query-keyword sweeps. The time advantage thus comes with higher traffic in these regimes.

Vocabulary growth changes BRASP's time little and increases traffic slightly. Query-keyword traffic is nearly constant, with time rising at ten keywords after remaining nearly unchanged from two to eight. Keyword count alone cannot predict complete cost: conjunctive predicates can reduce candidates, and padding can fix scheduled counts. The payload $2s_Q(M+L_I)+2v_Q(R+L_O)$ includes candidate retrieval before exact filtering.

Table~\ref{tab:search_breakdown_appendix} reports additional configurations, with latency increasing from 61 to 278\,ms and communication from 0.8 to 6.5\,MB. Database size, vocabulary, and query-keyword count vary jointly. These rows are compared within the table; they neither isolate one factor nor decompose the figure's measurements. Shared object counts alone do not align workloads or timing boundaries.

\begin{figure}[t]
  \centering
  \centering
    \begin{minipage}[t]{0.32\linewidth}
        \centering
        \includegraphics[width=\linewidth]{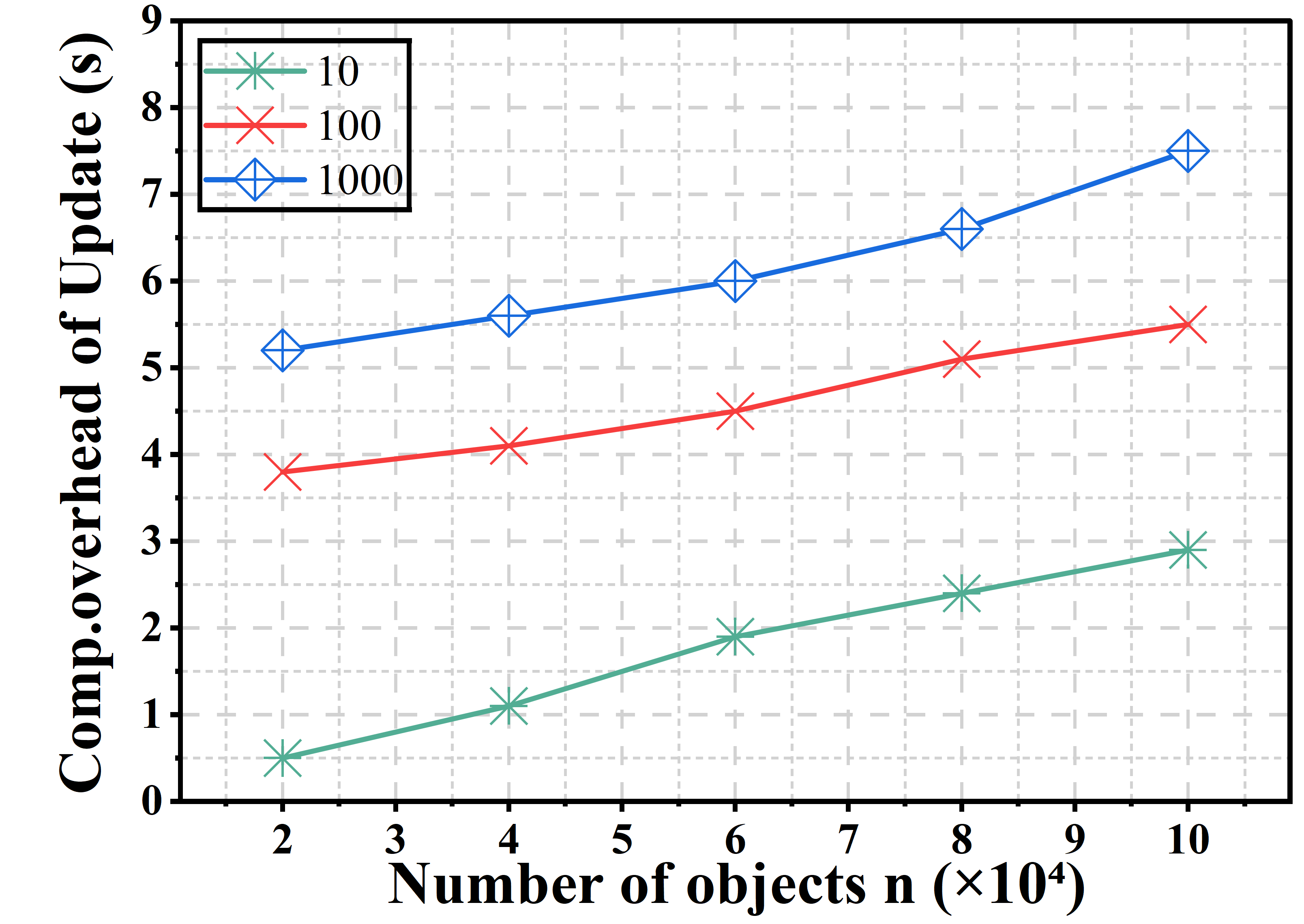}
        \subcaption{}
        \label{up1}
    \end{minipage}\hspace{0.015\linewidth}%
    \begin{minipage}[t]{0.32\linewidth}
        \centering
        \includegraphics[width=\linewidth]{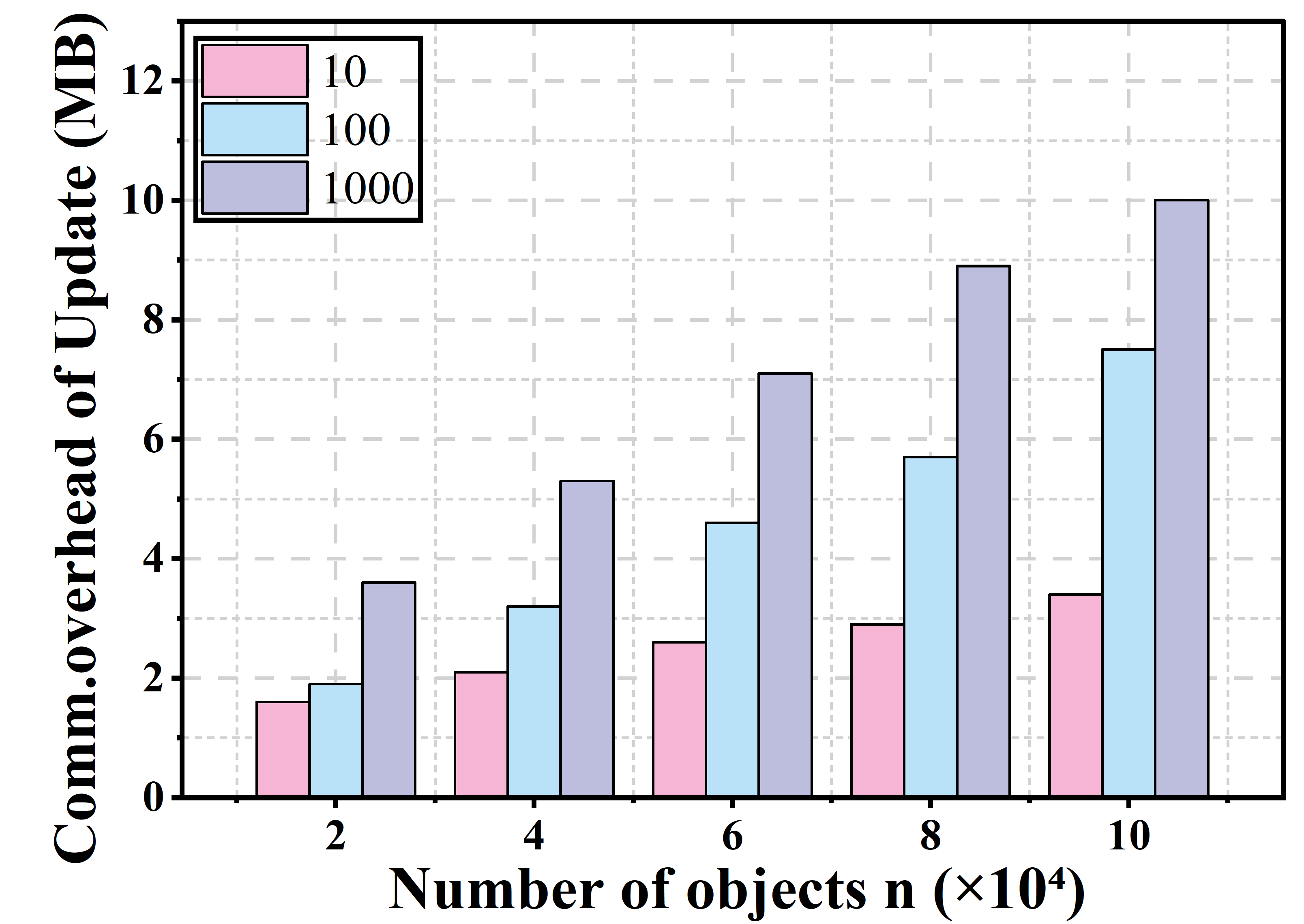}
        \subcaption{}
        \label{up2}
    \end{minipage}
  \caption{Recorded insertion-workload costs versus database size for
  $w_o\in\{10,100,1000\}$: (a) computation; (b) communication.}
  \label{update}
\end{figure}

\paragraph{\textbf{Insertion workloads.}}
In Fig.~\ref{update}, time and traffic grow with database size for $w_o\in\{10,100,1000\}$, and larger insertion workloads cost more at each size. This metric does not measure a complete snapshot replacement. At fixed capacities, $k$ commits transfer $3kD$ bits; a batch of $b$ shares $3D/b$ per insertion while delaying visibility until activation. The plotted $w_o$ alone does not specify this commit schedule.

\begin{table}[!htbp]
\caption{Measured encrypted index storage across database sizes at $m=500$.}
\label{tab:storage_memory_appendix}
\centering
\resizebox{\linewidth}{!}{
\begin{tabular}{c c c c}
\toprule
$n$ &
CS1 Index (MB) & CS2 Index (MB) & Total Index Size (MB) \\
\midrule
$2\times 10^4$  & 13.2  & 13.2  & 26.4 \\
$4\times 10^4$  & 26.5  & 26.5  & 53.0 \\
$6\times 10^4$  & 39.9  & 39.9  & 79.8 \\
$8\times 10^4$  & 53.4  & 53.4  & 106.8 \\
$10\times 10^4$ & 67.0  & 67.0  & 134.0 \\
$15\times 10^4$ & 100.8 & 100.8 & 201.6 \\
$20\times 10^4$ & 134.7 & 134.7 & 269.4 \\
\bottomrule
\end{tabular}
}
\end{table}

\paragraph{\textbf{Storage.}}
At vocabulary size 500, combined encrypted index storage grows from 26.4 to 269.4\,MB as the database scales tenfold to 200,000 objects (Table~\ref{tab:storage_memory_appendix}). Both replicas have equal footprints. This 10.2-fold growth excludes object records, client metadata, and staging (Section~\ref{theoretical_analysis}); query efficiency must therefore be weighed against replicated storage and full-snapshot maintenance.

\paragraph{Scope and limitations.}
BRASP assumes non-colluding servers, private channels, and a trusted online coordinator; it provides no restricted-query authorization against that client. Public volumes and timing permit inference. Padding costs extra and cannot use secret-dependent overflow fallbacks. Local side channels, malicious servers, deletion, ranked queries, and concurrent writers are outside scope. Appendix~\ref{app:limitations} discusses operational limitations.

\subsection{Implications of the Measured Tradeoffs}
The construction and search measurements favor BRASP in recorded time, while the search-communication plots expose a resource tradeoff. For deployment, network transfer and server scan work must be evaluated alongside local computation. In particular, smaller request-generation measurements do not replace the two-stage payload accounting in Eq.~\eqref{eq:query-cost}. The scheduled index and object reads determine how much data the complete protocol exchanges, and candidate retrieval is charged before the final answer is known.

The storage and insertion results concern complementary costs. Table~\ref{tab:storage_memory_appendix} measures encrypted indexes at both servers. A complete deployment also stores object records, private client directories, and temporary snapshots during commits. Similarly, an insertion-workload size $w_o$ describes the recorded workload, while the number of completed commits determines the $3kD$ maintenance payload. Relating the measurements to these quantities makes the design choice explicit: read-only queries avoid query-triggered maintenance, and full replacement pays for insertion privacy at public commit boundaries.

\FloatBarrier
\section{Conclusion}
\label{conclusion}
BRASP combines private bitmap and object retrieval with local exact filtering and complete insertion commits. Its correctness and transcript privacy cover these operations jointly under the non-colluding semi-honest model and explicit volume/scheduling leakage. Yelp measurements show lower construction and search times than the evaluated baselines, with higher query communication in several regimes. Full scans, replicated storage, and snapshot replacement remain the principal resource costs. The design separates read-only query execution from insertion visibility: queries reuse a committed snapshot, while batching shares maintenance work until the next epoch becomes available.

\bibliographystyle{splncs04}
\bibliography{mybibliography}

\appendix

\section{Security Proofs}
\label{app:security-proofs}

\subsection{Correctness Invariants}
Every committed epoch satisfies four invariants: (i) live objects have distinct stable IDs and mask $V_e$; (ii) $D_e$ maps each active predicate to its exact bitmap and absent predicates to $d_e$; (iii) $\pi_e$ maps live IDs to their correctly encoded records and ID $0$ to a dummy; (iv) both replicas contain identical ciphertexts bound to table, epoch, and physical slot. For a private read,
\[
 \bigoplus_{t:q_1[t]=1}C[t]\;\oplus\!\bigoplus_{t:q_2[t]=1}C[t]
 =\bigoplus_{t:(q_1\oplus q_2)[t]=1}C[t]=C[j].
\]
Correct AEAD decryption recovers the row; independently encrypted replicas would invalidate this identity.

The exact prefix cover selects precisely the covered cells. Every true spatial match lies in a covered cell, so Eq.~\eqref{eq:bitmap-query} retains every true answer and only permits boundary-cell spatial false positives after enforcing keyword membership. Privately fetching all candidates and checking original coordinates removes those false positives. Missing keywords contribute zero to the conjunction; missing prefixes contribute zero to the union. Padding responses are ignored; empty keyword/spatial expressions use $1^N/0^N$, and $V_e$ removes unused slots. These rules also cover empty answers.

BuildSnapshot and initial installation establish the invariants. A commit recovers $DB_e$, retains reserved insertion IDs, and constructs each bitmap from the entire extended database, including newly active predicates. Permutations change positions without changing membership. Both replicas and matching metadata become active together after outstanding reads finish. Induction over commits proves Theorem~\ref{thm:correctness}.

\subsection{Simulation of an Entire Adaptive History}
Fix corrupted server $i$. For a target $j$ chosen before the next fresh selection and any preceding history $H$,
\[
 \Pr[q_i=z\mid j,H]=2^{-m}\quad(z\in\{0,1\}^m).
\]
This is uniform sampling for $i=1$ and the bijection $z\mapsto z\oplus\mathbf e_j$ for $i=2$. It includes repeated/dummy targets and object targets derived from earlier results. Each local scan and answer is determined by the selection and stored table.

\paragraph{Logical execution.}
Throughout the hybrids, the challenger separately maintains the real database, directories, buffer, and plaintext outputs. This logical execution determines later adaptive inputs and leakage. Only the server-facing representation changes: substituted zero records are never used to compute the client's logical results.

\paragraph{Games $G_0$ and $G_1$.}
$G_0$ is the real retained view, including every snapshot copy, selection, answer, acknowledgment, and download. In $G_1$, replace each newly encrypted row payload by an equal-length zero string while preserving its table, epoch, physical slot, nonce, and associated data. Replacement occurs once at creation; subsequent reads, downloads, and replicas reuse that ciphertext. New snapshots are fully prepared before upload.

A multi-message left-or-right AEAD reduction maintains the logical execution and requests real or zero payloads with distinct nonces. From challenge ciphertexts it computes PIR answers and transfers without knowing the challenge key; the logical interpreter supplies real outputs for adaptive inputs. A hybrid over the two independent table keys bounds the distinguishing advantage by the sum of their confidentiality advantages. Zero encryptions are valid ciphertexts, not assumed uniform random strings.

In supported histories counters start at zero and each snapshot consumes exactly $M$ index and $R$ object nonces. No aborted preparations consume hidden values. Nonces and associated data therefore depend only on the disclosed dimensions, epochs, new physical ordering, and schedule, not on secret permutations or old identities. Equal row lengths cover real, empty, and dummy entries; the resulting snapshot distribution is independent of contents, occupancy, and logical names.

\paragraph{Game $G_2$.}
For each scheduled read, directly sample the corrupted server's uniform $m$-bit selection and compute its answer from the stored zero-payload table, omitting the hidden target and peer selection from view generation. The conditional identity above makes this change exact at each step, hence for the joint adaptive transcript. Logical outputs and leakage still come from the separate interpreter.

\paragraph{Simulator.}
On setup leakage, $S_i$ samples independent table keys and constructs encrypted-zero arrays with the prescribed nonces and associated data, retaining all versions. On $(e,s_Q,v_Q)$, it samples $s_Q$ independent $M$-bit and $v_Q$ independent $R$-bit selections and computes the corresponding answers from epoch $e$, following public headers and scheduling. On a commit it emits the complete old-table download only if $CS_i$ is the designated source, prepares new zero arrays, and reproduces its uploads, acknowledgments, and activation. Buffering emits nothing. It needs no predicate, affected address, or result ID. The ideal functionality maintains the real data and adaptive outputs and supplies leakage; the simulator's view is exactly $G_2$. Together with AEAD indistinguishability this proves Theorem~\ref{thm:transcript} under the stated semi-honest model.

\subsection{Correspondence and Forward Privacy}
For Proposition~\ref{prop:rebuild}, histories with identical disclosed schedules and counts but different hidden row correspondences produce the same encrypted-zero distribution in $G_1$. Full transfers contain no old-to-new pairs; nonces name only new physical positions. In $G_2$, later reads reveal only uniform selections and declared counts. Any additional correspondence test would distinguish a real view from its simulation.

For Corollary~\ref{cor:forward}, fix the preceding history and two supported batches with the same public size and commit metadata. Buffering emits nothing, and complete replacement reveals neither sparse destinations nor new-predicate markers. Their ideal views through the commit are identical, so their real views are computationally indistinguishable. Subsequent queries can disclose different candidate counts; the extension applies only when that additional leakage agrees. The claim is insertion-time forward privacy, not concealment of all later size-based inference.

\section{Operational Limitations}
\label{app:limitations}
The theorem covers uninterrupted installation. Persistent nonce consumption prevents reuse after abandoned work but does not prove recovery privacy: counter gaps and partial transfers need a failure-schedule leakage analysis. Reads scan whole tables; commits pause queries and require full download, workspace, and replicated upload. Batching delays visibility within fixed capacities. Finer Hilbert resolution changes prefix counts and boundary candidates, so it need not reduce total cost.

\section{Snapshot Construction}
\label{app:snapshot-construction}
\begin{algorithm}[H]
\caption{BuildSnapshot$(DB,e)$ at the Trusted Coordinator}
\label{alg:build-snapshot}
\KwIn{database $DB$ with distinct logical IDs in $\{1,\ldots,N\}$; target epoch $e$;
implicit keys, $\mathsf{pp}$, and persistent nonce counters}
\KwOut{candidate arrays $C_e^I,C_e^O$ and metadata $D_e,d_e,\pi_e,V_e$; nonce counters advanced persistently}
Form $\mathcal X_e$ from the records; initialize $B_x\leftarrow0^N$ for each
$x\in\mathcal X_e$ and $V_e\leftarrow0^N$\;
\ForEach{occupied logical identifier $j$}{
    $V_e[j]\leftarrow1$\;
    Set bit $j$ in every $B_{(\mathsf{keyword},w)}$ with $w\in W_j$\;
    Set bit $j$ in every $B_{(\mathsf{prefix},p)}$ with
    $p\in\mathcal P(X_j)$\;
}
Create an $M$-slot logical index containing these bitmaps, one designated zero
row, and zero rows in all remaining slots\;
Sample an independent uniform permutation of the $M$ slots; record its active
name mapping $D_e$ and designated zero-row position $d_e$ privately\;
Sample an independent uniform bijection
$\pi_e:\{0,\ldots,N\}\rightarrow\{1,\ldots,R\}$\;
Place each occupied object's padded encoding in slot $\pi_e(j)$ and dummy
encodings in all other object slots\;
\ForEach{table type $T\in\{\mathsf I,\mathsf O\}$}{
    \ForEach{physical slot $t$ in increasing order}{
        Allocate a new nonce $\nu$ from the persistent counter for $K_T$\;
        Encrypt the complete padded plaintext row using associated data
        $(T,e,t)$, and store $\nu$ with the ciphertext\;
    }
}
\Return the candidate arrays and metadata without replacing the active snapshot; publish no logical directory\;
\end{algorithm}

\FloatBarrier
\section{Extended Context and Complementary Objectives}
\label{app:extended-related}

\subsection{Cryptographic Transformations and Snapshot Replacement}
Paillier's construction develops homomorphic public-key encryption from composite residuosity~\cite{paillier1999public}. Atomic proxy cryptography uses a proxy key to transform ciphertexts for one key into ciphertexts for another~\cite{blaze1998divertible}. Universal re-encryption permits ciphertext rerandomization without knowing the public key and supports mixnet constructions~\cite{golle2004universal}. These mechanisms expose different interfaces for processing or transforming encrypted data. In BRASP, private reads combine ciphertext rows by XOR across identical replicas, and the trusted coordinator decrypts and reconstructs the full database at a commit. Fresh encryption and private permutations then produce the replacement snapshot. The security argument concerns that complete replacement and the two-server read transcript; it does not invoke proxy conversion or a universal re-encryption primitive.

\subsection{Other Query Operators and Service Guarantees}
Encrypted query processing also includes operators over relations and graphs. Chang et al.~\cite{chang2023towards} study oblivious general band joins and multiway equi-joins, accounting for the accesses made through their indexes. Ge et al.~\cite{ge2023privacy} investigate privacy-preserving graph matching with quick subgraph extraction. PM$k$R~\cite{xu2025pmkr} seeks nearby reachable graph vertices that satisfy multiple keywords. These tasks use relationships among tuples or vertices in addition to individual record attributes. BRASP's predicate instead combines a spatial rectangle with keyword conjunction over individual objects. Supporting joins, graph topology, or ranked reachability would require additional query operators and an analysis of their intermediate results and observable access counts.

Storage integrity and communication authentication are further service objectives. Chen et al.~\cite{chen2025popularity} study popularity-dependent replica storage for cloud integrity auditing, balancing storage efficiency and data availability. PrivHome~\cite{poh2019privhome} integrates authentication, protected storage, and encrypted queries for smart-home data. Such objectives complement the privacy of a storage transcript. BRASP assumes private authenticated channels and non-colluding semi-honest replicas; its identical replicas support private retrieval, while its theorem does not establish malicious-server auditing or availability. An integrated service must specify these guarantees separately and account for the messages introduced by their composition.

\end{document}